\documentclass{article}
\usepackage{ijcai19}

\usepackage{times}
\usepackage{soul}
\usepackage{url}
\usepackage[utf8]{inputenc}
\usepackage[small]{caption}
\usepackage{graphicx}
\usepackage{amsmath}
\usepackage{booktabs}
\usepackage{algorithm}
\usepackage{algorithmic}
\usepackage{amssymb}
\usepackage{amsmath}
\usepackage{tikz}
\usetikzlibrary{calc}
\usepackage{todonotes}
\usepackage{dsfont}
\usepackage{enumitem}
\usepackage{tikz}
\usepackage{amsthm}
\usepackage{thm-restate}
\usepackage{multirow}
\usepackage{amsfonts}
\usepackage{mathtools}
\usepackage{wrapfig}
\newtheorem{theorem}{Theorem}
\newtheorem{corollary}{Corollary}[theorem]
\newtheorem{lemma}{Lemma}
\newtheorem{definition}{Definition}
\newtheorem{example}{Example}

\newtheorem{remark}{Remark}

\newcommand{\X}{\mathcal{X}}
\newcommand{\bx}{\mathbf{x}}
\newcommand{\bX}{\mathbf{X}}
\newcommand{\Or}{\mathcal{O}}

\title{Be a Leader or Become a Follower: \\The Strategy to Commit to with Multiple Leaders \\ (Extended Version)}

\author{
	Matteo Castiglioni$^1$\footnote{Equal Contribution}
	\and
	Alberto Marchesi$^{1\ast}$
	\And
	Nicola Gatti$^1$
	\affiliations
	$^1$ Politecnico di Milano, Piazza Leonardo da Vinci 32, Milano, Italy 
	\emails
	\{matteo.castiglioni, alberto.marchesi, nicola.gatti\}@polimi.it
}

\begin{document}

\maketitle

\begin{abstract}
We study the problem of computing \emph{correlated strategies to commit to} in games with \emph{multiple} leaders and followers. To the best of our knowledge, this problem is widely unexplored so far, as the majority of the works in the literature focus on games with a single leader and one or more followers. The fundamental ingredient of our model is that a leader can decide whether to participate in the commitment or to defect from it by taking on the role of follower. This introduces a preliminary stage where, before the underlying game is played, the leaders make their decisions to reach an \emph{agreement} on the correlated strategy to commit to. We distinguish three solution concepts on the basis of the constraints that they enforce on the agreement reached by the leaders. Then, we provide a comprehensive study of the properties of our solution concepts, in terms of existence, relation with other solution concepts, and computational complexity.
\end{abstract}

\section{Introduction}\label{sec:intro}

%
%

Over the last years, Stackelberg games are receiving an increasing attention from the algorithmic game theory community, thanks to their many applications in real-world scenarios, such as in security~\cite{tambe2011security}.
In the classical Stackelberg setting~\cite{Stackelberg34:Marktform}, there is a leader with the ability to play before the other player, who acts as follower by observing the realization of the leader's strategy.
In this work, we follow a different approach, where the leader looks for a \emph{strategy to commit to}~\cite{conitzer2006computing}, and the follower observes the leader's mixed strategy, without knowing its actual realization.
An interpretation of this setting is provided by~\citeauthor{von2010leadership}~[\citeyear{von2010leadership}]: any (\emph{underlying}) game is extended as a sequential game in which the leader plays first, having a continuum of choices corresponding to mixed-strategy commitments.
%

The majority of the works in the literature focus on games with a single leader and a single follower~\cite{conitzer2006computing,von2010leadership}.
In this setting, the leader seeks for a utility-maximizing mixed strategy to commit to, while the follower plays a best response to the commitment.
This model has been largely studied, especially for security applications~\cite{paruchuri2008playing,KiekintveldJTPOT09,an2011guards}.

Some works also study games with a single leader and multiple followers.
\citeauthor{conitzer2011commitment}~[\citeyear{conitzer2011commitment}] introduce a model where the leader commits to a \emph{correlated strategy} and, accordingly, she draws recommendations for the followers, who must obey the incentive constraints of correlated equilibria~\cite{aumann1974subjectivity}.
The authors show that, in normal-form games, an optimal correlated strategy to commit to can be computed in polynomial time.
Other works study situations where the followers play a \emph{Nash equilibrium}~\cite{nash1951non} in the game resulting from the leader's mixed-strategy commitment~\cite{von2010leadership,coniglio17pessimistic,algo2018computing,de2018computing,Marchesi18:leadership,ijcai19cong}.
However, these models usually result in intractable computational problems even with a fixed number of followers.

Settings including multiple leaders are widely unexplored in the literature.
In spite of this, many real-world applications naturally involve more than one player with competitive advantages, playing the role of leader.
Some scenarios are, \emph{e.g.}, network platforms with premium (prioritized) users, markets where a group of firms forms a price-determining dominant cartel
, and political elections in which some candidates choose policy positions in advance of challengers.

Restricted to the security context, there are some works addressing games with multiple uncoordinated defenders (leaders)~\cite{smith2014multidefender,lou2015equilibrium,laszka2016multi,lou2017multidefender,gan2018stackelberg}.  
However, differently from our work, they all enforce Nash-like constraints on the leaders' strategies.
Moreover, their models suffer from two major drawbacks: (i) an exact equilibrium may not exist, and (ii) they strongly rely on problem-specific structures arising in security problems. 

The operations research literature provides further works on multi-leader-follower settings, under the name of \emph{mathematical programs with equilibrium constraints}~\cite{luo1996mathematical}.
%
%
They assume that both leaders and followers are subject to Nash constraints, with the latter playing in the game resulting from the leaders' strategies~\cite{leyffer2010solving,kulkarni2014shared}.
%
%
Furthermore, other works from the same field focus on oligopoly models where the leaders select the level of investment to maximize profits~\cite{demiguel2009stochastic}.
%
%
All these works considerably depart from ours, as they use fundamentally different models and lack thorough game theoretic and computational studies.

\subsubsection{Original Contributions}
%
%
We introduce a new way to apply the Stackelberg paradigm to any finite (underlying) game.
Our approach extends the idea of commitment to \emph{correlated strategies} in settings involving multiple leaders and followers, generalizing the work of \citeauthor{conitzer2011commitment}~[\citeyear{conitzer2011commitment}].
%
%
The crucial component of our framework is that a leader can decide whether to participate in the commitment or to defect from it by becoming a follower. 
This induces a preliminary \emph{agreement stage} that takes place before the underlying game is played, where the leaders decide, in turn, whether to {opt out} from the commitment or not.
%
We model this stage as a sequential game, whose size is factorial in the number of players.
%
%
%
%
%
%
%
%
%
%
Our goal is to identify commitments guaranteeing some desirable properties that we define on the agreement stage.
%
The first one requires that the leaders do not have any incentive to become followers.
It comes in two flavors, called \emph{stability} and \emph{perfect stability}, which are related to, respectively, Nash and subgame perfect equilibria of the sequential game representing the agreement stage.
The second property is also defined in two flavors, namely \emph{efficiency} and \emph{perfect efficiency}, both enforcing Pareto optimality with respect to the leaders' utility functions, though at different levels of the agreement stage. 

We introduce three solution concepts, which we generally call  \emph{Stackelberg correlated equilibria} (SCEs).
They differ depending on the properties they call for.
Specifically, SCEs, SCEs \emph{with perfect agreement} (SCE-PAs), and SCE-PAs \emph{with perfect efficiency} (SCE-PAPEs) require, respectively, stability and efficiency, perfect stability and efficiency, and both perfect stability and perfect efficiency.

First, we investigate the game theoretic properties of our solution concepts.
We show that SCEs and SCE-PAs are guaranteed to exist in any game, while SCE-PAPEs may not.
Moreover, we compare them with other solution concepts.

Then, we switch the attention to the computational complexity perspective.
We show that, provided a suitably defined \emph{stability oracle} is solvable in polynomial time, an SCE optimizing some linear function of leaders' utilities (such as the leaders' social welfare) can be computed in polynomial time, even in the number of players.
%
The same holds for finding \emph{an} SCE-PA, while we prove that computing an optimal SCE-PA is an intractable problem. 
Nevertheless, in the latter case, we provide an (exponential in the game size) upper bound on the necessary number of queries to the oracle.
%

In conclusion, we study which classes of games admit a polynomial-time stability oracle, focusing on those with \emph{polynomial type}~\cite{papadimitriou2008computing}.
We show that the problem solved by our oracle is strictly connected with the weighted deviation-adjusted social welfare problem introduced by~\citeauthor{leyton2011ellipsoid}~[\citeyear{leyton2011ellipsoid}].
As a result, we get that our oracle is solvable in polynomial time in all the game classes where the same holds for the problem of finding an optimal correlated equilibrium.~\footnote{Full proofs of all the results are in Appendices~\ref{sec:props_existence_app},~\ref{sec:prop_relations_app},~\ref{sec:computational_app},~and~\ref{sec:games}.}
%

\section{Preliminaries}\label{sec:prelim}

In this section, we introduce some basic concepts about games and their equilibria used in the rest of the paper.
%

\subsection{Finite Games and Their Equilibria}

A \emph{(finite) game} $G$ is a tuple $(N,\{S_p\}_{p \in N}, \{u_p\}_{p \in N})$, where $N = \{1,\ldots,n\}$ is a finite set of players, $S_p$ is a finite set of player $p$'s \emph{strategies} or \emph{actions}, and $u_p : S \to \mathbb{R}$ is player $p$'s \emph{utility}, defined over the set of \emph{strategy profiles} $S = \bigtimes_{p \in N} S_p$.
Given $s \in S$, let $s_{-p}\in S_{-p} =  \bigtimes_{q \in N \setminus \{p\}} S_q$ be the partial profile obtained by dropping player $p$'s strategy $s_p$ from $s$, so that $s = (s_p, s_{-p})$.
We call $\X = \Delta(S)$ the set of \emph{correlated distributions} defined over strategy profiles, \emph{i.e.}, each $x \in \X$ satisfies $\sum_{s \in S} x(s) = 1$ and $x(s) \geq 0$ for all $s \in S$.
Moreover, overloading notation, $u_p(x) =\sum_{s \in S} x(s) u_p(s)$ is player $p$'s expected utility in $x \in \X$.

A correlated distribution $x \in \X$ is a \emph{correlated equilibrium} (CE)~\cite{aumann1974subjectivity} if, for every player $p \in N$ and strategies $s_p \neq s_p' \in S_p$, the following constraint holds: 
\begin{equation}\label{eq:incentive_ce}
	\sum_{s_{-p} \in S_{-p}} x(s_p,s_{-p}) \left( u_p(s_p,s_{-p}) - u_p(s_p',s_{-p})  \right) \geq 0.
\end{equation}
%
We can interpret a CE in terms of a mediator who draws some strategy profile $s \in S$ from a publicly known distribution $x$, and, then, it privately communicates each recommendation $s_p$ to player $p$.
The distribution is an equilibrium if no player has an incentive to deviate from the recommendation, as made formal by the \emph{incentive constraints} of Eq.~\eqref{eq:incentive_ce}.
%
%
%
%
%
Moreover, a \emph{Nash equilibrium} (NE)~\cite{nash1951non} is a CE $x \in \X$ that can be written as the product distribution of players' \emph{mixed strategies}, \emph{i.e.}, $x(s) = \prod_{p \in N} x_p(s_p)$ for all $s \in S$, where each $x_p\in \Delta(S_p)$ is a probability distribution over strategies $S_p$ denoting a player $p$'s mixed strategy.

In the following, we denote with $\X^\textsc{CE}_P \subseteq \X$ the set of correlated distributions that satisfy the incentive constraints of Eq.~\eqref{eq:incentive_ce} only for a subset of players $P \subseteq N$.
Clearly, $\X^\textsc{CE} = \X^\textsc{CE}_N$ is the set of CEs of the game.
%

Different classes of games are employed depending on how strategies and utilities are represented.
The most common representation is the \emph{normal form}, which encodes each utility function $u_p$ as an $n$-dimensional matrix indexed by $s \in S$.
%
Thus, the size of a normal-form game is exponential in the number of players.
Many other representations have been introduced in the literature.
In Section~\ref{sec:games_main}, we are interested in those with \emph{polynomial type}~\cite{papadimitriou2008computing}, where the number of players and the number of strategies are bounded by polynomials in the size of the game.
%
%
Many important classes of games admit a polynomial-type representation, such as graphical games~\cite{kearns2013graphical}, polymatrix games~\cite{eaves1973polymatrix}, anonymous games~\cite{blonski2000characterization}, and congestion games~\cite{rosenthal1973class}.

\subsection{Stackelberg Games and Equilibria}

Any finite game has a Stackelberg counterpart where some of the players are \emph{leaders} and the others are \emph{followers}.
The former have the ability to commit to a course of play beforehand, while the latter decide how to play after observing the commitment~\cite{von2010leadership}.

\begin{definition}\label{def:stackelberg_game}
	Given a finite game $G$, a \emph{Stackelberg game (SG)} is a tuple $(G,L,F)$ where $L$ and $F$ are the sets of \emph{leaders} and \emph{followers}, respectively, with $N = L \cup F$.
\end{definition}

%
In single-leader single-follower SGs, the follower best responds to the leader's mixed-strategy commitment~\cite{conitzer2006computing,von2010leadership}.
%
%
%
\begin{definition}\label{def:stackelberg_eq}
	Given an SG $(G, \{1\}, \{2\})$, a leader's mixed strategy $x_1 \in \Delta(S_1)$ defines a \emph{Stackelberg equilibrium (SE)} if it maximizes $u_1$ given that, for each $x'_1 \in \Delta(S_1)$, the follower plays an $x_2(x'_1) \in \Delta(S_2)$ maximizing $u_2$.~\footnote{In the literature, different SE concepts are defined depending on how the follower is assumed to break ties. The \emph{strong} and \emph{weak} SEs are two notable cases~\cite{breton1988sequential}, where the follower is assumed to break ties in favor and against the leader, respectively.} 
\end{definition}

%
The multi-follower case unfolds in different scenarios depending on how the followers are assumed to play.
%
%
\citeauthor{conitzer2011commitment}~[\citeyear{conitzer2011commitment}] study what they call \emph{optimal correlated strategies to commit to}, where the leader commits to a utility-maximizing correlated distribution satisfying the incentive constraints (Eq.~\eqref{eq:incentive_ce}) for the followers only.
Formally:

\begin{definition}\label{def:commit_correlated}
	Given an SG $(G,\{ 1 \}, N\setminus \{ 1 \})$, $x \in \X$ is an \emph{optimal correlated strategy to commit to} if it maximizes the leader's utility $u_1(x)$ over the set $\X^\textsc{CE}_{N \setminus \{1\}}$.
\end{definition}

In our work, we pursue the approach of~\citeauthor{conitzer2011commitment}~[\citeyear{conitzer2011commitment}], rather than letting the followers play an NE, as done, \emph{e.g.}, by~\citeauthor{von2010leadership}~[\citeyear{von2010leadership}].    
Indeed, while the two models provide the same leader's utility in single-follower SGs (corresponding to that in an SE), the latter may be strictly better in SGs with two or more followers (see~\cite{conitzer2011commitment} for an example).

\section{Multi-Leader-Follower Stackelberg Games}

%
%
%
We address SGs with multiple leaders and followers.
%
%
%
%
The key components of our approach are the following.
First, we allow the leaders to decide whether to participate in the commitment or to defect from it by taking on the role of followers.
This is modeled by the \emph{agreement stage} of the SG, whose result is the formation of an {agreement} involving a subset of the leaders.
Second, in the spirit of CEs, we introduce a correlation device that, after the agreement, draws recommendations and privately communicates them to the players.
Following~\citeauthor{conitzer2011commitment}~[\citeyear{conitzer2011commitment}], we assume that the leaders involved in the agreement commit to play their recommendations, while the followers obey to the usual incentive constraints of CEs (see Eq.~\eqref{eq:incentive_ce}).
The correlation device may adopt different distributions depending on the sequence of defections that determined the agreement, and these distributions are publicly known.
Our goal is to design the device, so as to achieve some desirable properties of the commitment, which we formally describe in the rest of the section.

Before going into our main definitions, we introduce some useful notation.
Given a subset of players $P \subseteq N$, we denote with $\Pi_P$ the collection of ordered subsets of $P$, including the empty set $\varnothing$.
Given $\pi \in \Pi_P$ and $p \in P \setminus \pi$, we let $\pi p$ be the ordered set obtained by appending $p$ at the end of $\pi$.
We use $\bx = [x_\pi]$ to denote a vector of correlated distributions $x_\pi \in \X^\textsc{CE}_{\pi \cup F}$, one per ordered subset of leaders $\pi \in \Pi_L$,
%
while $\bX = \bigtimes_{\pi \in \Pi_L} \mathcal{X}^\textsc{CE}_{\pi \cup F}$ is the set of all such vectors.
In words, $\pi \in \Pi_L$ represents a sequence of leaders' defections in the agreement stage, while $\bx$ defines the publicly known correlated distributions adopted by the correlation device, with $x_\pi$ being the one used when the sequence of defections is $\pi$.

\begin{definition}\label{def:stackelberg_game_for_x}
	Given a vector of distributions $\bx = [x_\pi] \in \bX$, an SG $(G,L,F)$ is structured in the following two stages:
	\begin{itemize}[nolistsep,itemsep=0mm]
		\item \textbf{Agreement.} It goes on in rounds. In a given round, each leader, in turn, decides between \textsc{Opt-In} and \textsc{Opt-Out}.\footnote{We assume that the leaders are asked to take a decision according to some ordering, \emph{e.g.}, $p \in L$ decides before $q \in L$ if $p < q$.}
		All the decisions are perfectly observable.
		If a player chooses \textsc{Opt-Out}, then she leaves the set of leaders becoming a follower, and a new round starts.
		The stage ends when, during a round, all remaining leaders decide to \textsc{Opt-In}.
		The result is the ordered subset $\pi\in \Pi_L$ of leaders who decided to \textsc{Opt-Out}.~\footnote{The agreement stage is finite 
			as there are at most $|L|$ rounds and each round involves at most $|L|$ decisions.
		}~\footnote{Our results do not rely on the protocol implemented in the agreement stage. Others could be adopted, with the only requirement that they must record in which order the leaders do \textsc{Opt-Out}.}
		%
		%
		%
		%
		%
		%
		%
		%
		%
		\item \textbf{Play.} The correlation device draws some $s \in S$ according to the publicly known correlated distribution $x_\pi$. Then, each player is privately told her recommendation and the underlying game $G$ is played, with the leaders in $L \setminus \pi$ sticking to their recommendations.
	\end{itemize}
\end{definition}

\begin{remark}
	The agreement stage of an SG can be represented as a sequential (i.e., tree-form) game involving the leaders.
	In such game, the players play in turn, according to some fixed order, with only two actions available at each decision point: \textsc{Opt-In} and \textsc{Opt-Out}.
	When a player chooses \textsc{Opt-Out}, then she never plays anymore.
	The game ends after a sequence of \textsc{Opt-In} actions performed by all leaders who have not selected \textsc{Opt-Out} yet.
	Thus, each leaf of the game corresponds to the ordered subset $\pi \in \Pi_L$ representing the sequence of leaders who performed \textsc{Opt-Out} on the path to the leaf.
	Players' payoffs are defined by $u_p(x_\pi)$ for $p \in L$.
	See Figure~\ref{table:example}~(Right) for an example of sequential-game-representation of the agreement stage.
\end{remark}

Next, we introduce some desirable properties that the distributions of the correlation device should satisfy.
In the following definitions, we assume that an SG $(G,L,F)$ is given.

First, we introduce \emph{stability}.
In words, we require that the leaders in $L$ do not have any incentive to become followers.
We introduce two different notions of stability, as follows.

\begin{definition}\label{def:stability}
	Given $\bx = [x_\pi] \in \bX$, for any $\pi \in \Pi_L$, $x_\pi$ is \emph{stable} if, for every $p \in L \setminus \pi$, $u_p(x_\pi)\geq u_p(x_{\pi p})$. Moreover:
	\begin{itemize}[nolistsep,itemsep=0mm]
		\item $\bx$ is \emph{stable} if $x_\varnothing$ is stable;
		\item $\bx$ is \emph{perfectly stable} if $x_\pi$ is stable for every $\pi \in \Pi_L$.
	\end{itemize}
\end{definition}

We denote with $\bX^\textsc{S} \subseteq \bX$ and $\bX^\textsc{PS} \subseteq \bX$ the sets of stable and perfectly stable distributions, respectively.

%
%

\begin{remark}
	The rationale behind stability is that of NE.
	Indeed, $\bx \in \bX$ is stable if and only if each leader playing \textsc{Opt-In} is an NE of the sequential game representing the agreement stage.
	Intuitively, this is because, if $\bx \in \bX$ is stable, each leader must not have any incentive to play \textsc{Opt-Out} given that the other leaders always play \textsc{Opt-In}.
\end{remark}

%
%
%

\begin{remark}
	The rationale behind perfect stability is that of subgame perfection.
	Indeed, $\bx \in \bX$ is perfectly stable if and only if each leader playing \textsc{Opt-In} is a subgame perfect equilibrium of the agreement stage.
	The reason is that perfect stability requires that playing \textsc{Opt-In} is optimal at any decision point of the sequential game.
\end{remark}

The second property that we look for is \emph{efficiency}.
We require that the correlated distributions of the correlation device are \emph{Pareto optimal} with respect to the utility functions of the leaders who decided to \textsc{Opt-In}.
Given $\bX' \subseteq \bX$, for $\pi \in \Pi_L$, we use $\mathcal{P}_{L \setminus \pi}(\bX')$ to denote the set of Pareto optimal correlated distributions in the set $\{ x'_\pi \mid \bx' = [x'_\pi] \in \bX' \}$, where the objectives are the functions $u_p$, for $p \in L \setminus \pi$.
Formally:

\begin{definition}\label{def:efficiency}
	Given $\bx = [x_\pi] \in \bX' \subseteq \bX$, for any $\pi \in \Pi_L$, $x_\pi$ is \emph{efficient} on the set $\bX'$ if $x_\pi \in \mathcal{P}_{L \setminus \pi}(\bX')$.
	Moreover:
	\begin{itemize}[nolistsep,itemsep=0mm]
		\item $\bx$ is \emph{efficient} on $\bX'$ if $x_\varnothing$ is efficient on $\bX'$;
		\item $\bx$ is \emph{perfectly efficient} on $\bX'$ if $x_\pi$ is efficient on $\bX'$ for every $\pi\in\Pi_L$.
	\end{itemize} 
\end{definition}

%
%
%
%

We introduce three different solution concepts for our SGs, which we refer to as \emph{Stackelberg correlated equilibria} (SCEs).
They differ on the types of stability and efficiency that they prescribe.
Formally:

\begin{definition}\label{def:sce}
	Given an SG $(G,L,F)$, $\bx \in \bX$ is an:
	\begin{itemize}[nolistsep,itemsep=0mm]
		\item \emph{SCE} if it is efficient on the set $\bX^\textsc{S}$;
		\item \emph{SCE with perfect agreement (SCE-PA)} if it is efficient on the set $\bX^\textsc{PS}$;
		\item \emph{SCE with perfect agreement and perfect efficiency (SCE-PAPE)} if it is perfectly efficient on the set $\bX^\textsc{PS}$.
	\end{itemize}
\end{definition}

%
%
%
%
%

We denote with $\bX^\textsc{SCE}$, $\bX^\textsc{SCE-PA}$, and $\bX^\textsc{SCE-PAPE}$ the sets of SCEs, SCE-PAs, and SCE-PAPEs, respectively.



%

\begin{example}
	Consider the SG in Figure~\ref{table:example}, where $L = \{1,2\}$ and $F = \varnothing$.
	Let $\bx=[x_\pi]$ be such that $x_\varnothing(s_{1,1},s_{2,1})=1$, $x_{\{2\}}(s_{1,5},s_{2,1})=1$, and $x_\pi(s_{1,1},s_{2,2})=1$ for all the other $\pi \in \Pi_L$.
	%
	Clearly, $x_\pi \in \X_\pi^\textsc{CE}$ for all $\pi \in \Pi_L$.
	Moreover, being $x_\varnothing$ stable and Pareto optimal, $\bx$ is an SCE.
	Observe that, if player 2 performs \textsc{Opt-Out}, $\bx$ prescribes an irrational behavior to player 1, as $u_1(x_{\{2\}}) =0$, while she gets $1$ by doing \textsc{Opt-Out}.
	Thus, $\bx$ is not perfectly stable, as playing \textsc{Opt-In} must be optimal at any decision point of the agreement stage.
	For instance, $\bx'=[x'_\pi]$ with $x'_\varnothing(s_{1,2},s_{2,1})=1$ and $x'_\pi(s_{1,3},s_{2,2})=1$ for every other $\pi \in \Pi_L$ is an SCE-PA.
	However, notice that $\bx'$ is not an SCE-PAPE since $x'_{\{2\}}$ does not maximize player 1's utility.
	%
	Instead, $\bx''=[x''_\pi]$ with $x''_\varnothing(s_{1,4},s_{2,1})=1$, $x''_{\{2\}}(s_{1,3},s_{2,1})$, and $x''_\pi(s_{1,4},s_{2,2})=1$ for all the other $\pi \in \Pi_L$ is an SCE-PAPE.
\end{example}

\begin{figure}[!htp]
	\centering
	\begin{minipage}{0.125\textwidth}
		{\renewcommand{\arraystretch}{1.1}\setlength{\tabcolsep}{2pt}\begin{tabular}{c|c|c|}	
				\multicolumn{1}{c}{}&\multicolumn{1}{c}{$s_{2,1}$}  & \multicolumn{1}{c}{$s_{2,2}$}  \\\cline{2-3}
				$s_{1,1}$ & $5,0$ & $1,2$ \\\cline{2-3}
				$s_{1,2}$ & $4,1$ & $1,2$  \\\cline{2-3}
				$s_{1,3}$ & $2,1$ & $1,1$ \\\cline{2-3}
				$s_{1,4}$ & $3,2$ & $1,3$  \\\cline{2-3}
				$s_{1,5}$ & $0,0$ & $0,0$  \\\cline{2-3}
		\end{tabular}}
	\end{minipage}
	\begin{minipage}{0.35\textwidth}
		\centering
		\includegraphics[width=\textwidth]{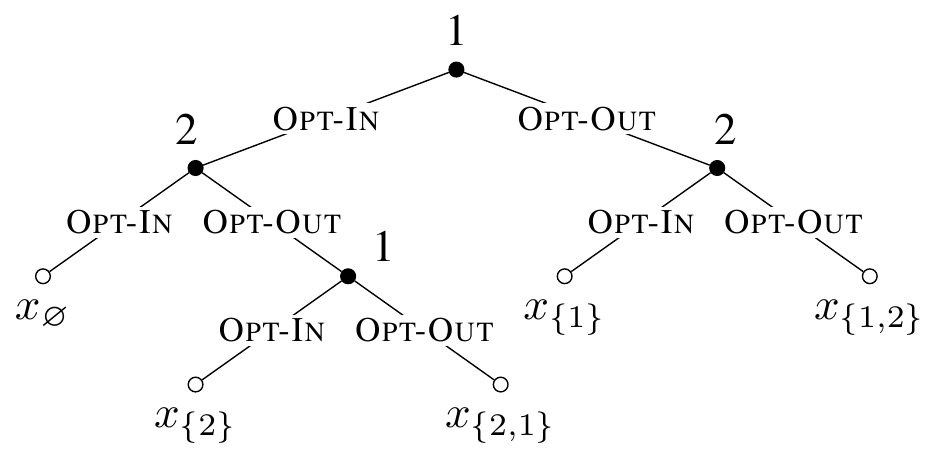}
	\end{minipage}
	\caption{\emph{Left}: Example of two-player normal-form SG with $L = \{1,2\}$. \emph{Right}: Sequential game representing its agreement stage.}
	\label{table:example}
\end{figure}

\section{On the Existence of SCEs}\label{sec:props_existence}

We investigate the existence of our solution concepts in general SGs.
We show that SCEs and SCE-PAs always exist, while we provide an SG where there is no SCE-PAPE.

The fundamental step for proving our existence results (Theorem~\ref{thm:existence_sce}) is to show that (i) $\bX^\textsc{S}$ and $\bX^\textsc{PS}$ are polytopes, and (ii) they are non-empty.
The latter point is a direct consequence of the fact that all vectors $\bx = [x_\pi] \in \bX$ with $x_\pi=x$ for some CE $x \in \X^\textsc{CE}$ are perfectly stable.

\begin{restatable}{theorem}{theoremtwo}
\label{thm:existence_sce}
	Every SG admits an SCE and an SCE-PA.
\end{restatable}


\begin{restatable}{prop}{propositionone}
	There are SGs with no SCE-PAPE.
	%
	%
	%
\end{restatable}

\begin{proof}[Proof sketch]
	Consider the SG in Table~\ref{table:no_sce_pape}, where $L=\{1,2,3\}$ and $F=\varnothing$.
	Any $\bx = [x_\pi] \in \bX^\textsc{SCE-PAPE}$ must be such that, for every $x_\pi$ with player 3 in $\pi$, $u_3(x_\pi) = 1$ (as player 3 always gets $1$ by deviating to $s_{3,3}$). 
	%
	%
	Given the definition of stability and player 3's incentive constraints, $x_{\{1,2\}}$ and $x_{\{2,1\}}$ must always recommend $s_{3,3}$ to player 3.
	%
	%
	Moreover, by stability and efficiency, $x_{\{1\}}$ must always recommend $(s_{1,1},s_{2,2},s_{3,2})$, where player 1 gets a utility of $2$.
	Similarly, $x_{\{2\}}$ must always recommend $(s_{1,2},s_{2,2},s_{3,1})$ and, thus, player 2 receives a utility of $2$.
	Thus, for stability, $x_\varnothing$ must satisfy $u_1(x_\varnothing), u_2(x_\varnothing) \geq 2$, which is impossible.
\end{proof}

As a result, in the rest of this work we focus on SCEs and SCE-PAs.
We remark that the non-existence of SCE-PAPEs implies that, under the requirements of perfect stability and perfect efficiency, there cannot be an agreement involving all the leaders.
This does not rule out the possibility that some subsets of leaders can still reach an agreement.
However, these cases are much more involved, as the actual group of leaders reaching an agreement inevitably depends on the rules of the protocol implemented in the agreement stage.

\begin{table}[t]
	\centering
	{\renewcommand{\arraystretch}{1.1}\setlength{\tabcolsep}{2pt}\begin{tabular}{cc|c|c|}
			& \multicolumn{1}{c}{} & \multicolumn{1}{c}{$s_{2,1}$}  & \multicolumn{1}{c}{$s_{2,2}$} \\\cline{3-4}
			& $s_{1,1}$ & $0,2,0$ & $2,0,0$ \\\cline{3-4}
			& $s_{1,2}$ & $0,2,0$ & $1,2,1$ \\\cline{3-4}
			& \multicolumn{1}{c}{} & \multicolumn{2}{c}{$s_{3,1}$}\\
	\end{tabular}}
	{\renewcommand{\arraystretch}{1.1}\setlength{\tabcolsep}{2pt}\begin{tabular}{cc|c|c|}
			& \multicolumn{1}{c}{} & \multicolumn{1}{c}{$s_{2,1}$}  & \multicolumn{1}{c}{$s_{2,2}$} \\\cline{3-4}
			& $s_{1,1}$ & $0,2,0$ & $2,1,1$ \\\cline{3-4}
			& $s_{1,2}$ & $0,2,0$ & $0,0,0$ \\\cline{3-4}
			& \multicolumn{1}{c}{} & \multicolumn{2}{c}{$s_{3,2}$}\\
	\end{tabular}}
	{\renewcommand{\arraystretch}{1.1}\setlength{\tabcolsep}{2pt}\begin{tabular}{cc|c|c|}	
			& \multicolumn{1}{c}{} & \multicolumn{1}{c}{$s_{2,1}$}  & \multicolumn{1}{c}{$s_{2,2}$} \\\cline{3-4}
			& $s_{1,1}$ & $0,0,1$ & $0,0,1$ \\\cline{3-4}
			& $s_{1,2}$ & $0, 0,1$ & $0,0,1$ \\\cline{3-4}
			& \multicolumn{1}{c}{} & \multicolumn{2}{c}{$s_{3,3}$}\\
	\end{tabular}}
	\caption{Three-player normal-form SG with no SCE-PAPE (players 1, 2, and 3 select rows, columns, and matrices, respectively).}
	\label{table:no_sce_pape}
\end{table}



%

\section{SCEs and Other Solution Concepts}\label{sec:prop_relations}

%
%
We show that the optimal correlated strategies to commit to introduced by~\citeauthor{conitzer2011commitment}~[\citeyear{conitzer2011commitment}] are a special case of SCEs.
Intuitively, in single-leader SGs, efficiency is equivalent to the maximization of leader's utility, while stability does not enforce additional constraints on the commitment.

\begin{restatable}{theorem}{theoremthree}
	Given an SG $(G,\{ 1 \}, N\setminus \{ 1 \})$, it holds $\bX^\textsc{SCE}=\bX^\textsc{SCE-PA}=\bX^\textsc{SCE-PAPE}$ and, given some $\bx = [x_\pi] \in \bX^\textsc{SCE}$, $x_\varnothing$ is an optimal correlated strategy to commit to.
\end{restatable}

\begin{proof}[Proof sketch]
	Since the SG has only one leader (player 1), 
	$\bX^\textsc{S} = \bX^\textsc{PS}$,
	and, thus, $\bX^\textsc{SCE}=\bX^\textsc{SCE-PA}$.
	For the same reasons, 
	$\bX^\textsc{SCE-PA}=\bX^\textsc{SCE-PAPE}$.
	Moreover, Pareto optimality is the same as maximizing the leader's utility function $u_1$.
	Let $\bx = [x_\pi] \in \bX^\textsc{SCE}$ and assume, by contradiction, that $x_\varnothing$ is not an optimal correlated strategy to commit to.
	Then, there would be another $\hat x \in \X^\textsc{CE}_{N \setminus \{1\}}$ such that $u_1(\hat x) \geq u_1(x_\varnothing)$.
	However, replacing $x_\varnothing$ with $\hat x$ in $\bx$ would give us another $\hat \bx \in \bX^\textsc{S}$, contradicting the efficiency of $\bx$.
	%
\end{proof}

Given the relation between optimal correlated strategies to commit to and SEs in single-leader single-follower SGs:

\begin{corollary}\label{cor:corollary_se}
	Given an SG $(G,\{ 1 \}, \{ 2 \})$, any $\bx = [x_\pi] \in \bX^\textsc{SCE}$  is such that $u_1(x_\varnothing)$ is the leader's utility in an SE.
	%
\end{corollary}

For the relationships of SCEs with other non-Stackelberg solution concepts, we refer the reader to Appendix~\ref{sec:sce_and_cor}.

\section{Computational Complexity of SCEs}\label{sec:computational}

%
%
%
%
%
%

We study the computational complexity of SCEs and SCE-PAs in general SGs. 
We distinguish between the problem of finding \emph{an} equilibrium and that of computing an \emph{optimal} equilibrium, \emph{i.e.}, one maximizing a specific given linear function of leaders' utilities, such as the leader's social welfare.
We introduce the following formal definitions (problems \textsf{f-SCE-PA} and \textsf{o-SCE-PA($\lambda$)} are defined analogously for SCE-PAs).

\begin{definition}[\textsf{f-SCE}]\label{def:find_sce}
	Given an SG $(G,L,F)$, find an SCE.
\end{definition}
\begin{definition}[\textsf{o-SCE($\lambda$)}]\label{def:opt_sce}
	Given an SG $(G,L,F)$ and $\lambda = [\lambda_p] \in [0,1]^{|L|}$, find an SCE $\bx = [x_\pi] \in \bX^\textsc{SCE}$ maximizing the objective function $f_\lambda = \sum_{p \in L}\sum_{s \in S} \lambda_p u_p(s) x_\varnothing(s)$. 
\end{definition}

Let us remark that, in general, the size of a vector $\bx \in \bX$ is factorial in the number of players.
Thus, in the following, we assume that there is some compact representation for $\bx$.~\footnote{As we see next, for all our positive results we can safely assume that there is a compact representation for $\bx \in \bX$ (\emph{e.g.}, $\bx$ only requires a polynomial number of polynomially-sized distributions).}

We establish a tight connection between our problems and an auxiliary one, which is a generalization of the problem of finding an optimal CE.
%
In the rest of the section, we assume to have access to an oracle solving this auxiliary problem, which we call \emph{stability oracle}.
%
In Section~\ref{sec:games_main}, we then investigate for which games the oracle can be efficiently implemented.
%
%
%

%
%
%
%

\begin{definition}\label{def:oracle}
	A \emph{stability oracle} $\Or(G,{c},L,\{x_p\}_{p \in L' \subseteq L})$ is an algorithm that, given a finite game $G$, a coefficients vector ${c} = [c_p] \in [-1,1]^{n}$, a set of leaders $L \subseteq N$, and a collection of correlated distributions $x_p \in \X$ for $p \in L' \subseteq L$, returns an $x \in \X^\textsc{CE}_{N \setminus L}$ maximizing $ \sum_{p \in N} \sum_{s \in S} c_p  u_p(s) x(s)$ subject to the stability constraints, i.e., $u_p(x) \geq u_p(x_p)$ for all $p \in L'$.~\footnote{Note that, given a finite game $G$, $\Or(G,{c},\varnothing,\varnothing)$ returns an optimal CE $x \in \X^\textsc{CE}$ for the objective function defined by ${c} \in [0,1]^n$.}
	%
%
%
%
\end{definition}

%
%
%

In the following, we are interested in games where the stability oracle runs in polynomial time.
Thus, we assume that $\Or$ always returns a correlated distribution with size polynomial in the size of the game.~\footnote{Indeed, this assumption is not restrictive, as all the games we study in Section~\ref{sec:games_main} admit a poly-time oracle $\Or$ with this property.}
We also consider the decision form of the stability oracle, which reads as follows:

%
%
%
%

\begin{definition}\label{def:oracle_dec}
	The \emph{decision form} of a stability oracle $\Or$ is an algorithm $\Or^\textsc{d}(x,L,\{x_p\}_{p \in L' \subseteq L})$ that, given $x \in \mathcal{X}$, $L \subseteq N$, and $x_p \in \X$ for $p \in L' \subseteq L$, answers \textsc{Yes} if $x \in \X^\textsc{CE}_{N \setminus L}$ and $x$ satisfies the stability constraints, and \textsc{No} otherwise.
\end{definition}

%
%
%
%
%

In the following, given $L \subseteq N$ and $\lambda = [\lambda_p] \in [0,1]^{|L|}$, we let ${c}_\lambda = [c_{\lambda,p}] \in [0,1]^n$ be such that $c_{\lambda,p} = \lambda_p$ if $p \in L$, while $c_{\lambda,p}=0$ if not.
Moreover, given $p \in N$, we let ${c}_{p} = [c_{p,q}] \in [0,1]^n$ be such that $c_{p,p}=-1$ and $c_{p,q}=0$ for all $q \in N \setminus \{p\}$.
Note that $c_\lambda$ is the coefficients vector of the objective $f_\lambda$, while $c_p$ corresponds to minimizing $u_p$.

\subsection{Computing SCEs}

We show that, in games admitting a polynomial-time stability oracle, an optimal SCE can be computed in polynomial time.
Intuitively, \textsf{o-SCE}($\lambda$) is solved by $\bx = [x_\pi]$ computed as: $x_{\{p\}}=\Or(G,{c}_p,L \setminus \{p\},\varnothing)$ for $p \in L$, $x_\varnothing=\Or(G,{c}_\lambda,L, \{x_{\{p\}}\}_{p\in L})$, and $x_\pi=\Or(G,{c}_\lambda,\varnothing, \varnothing)$ for every other ordered subset $\pi\in \Pi_L$.  
Formally:

%
%
%
%

%
%
%
%

\begin{restatable}{theorem}{theoremsix}
\label{thm:sce_easy}
	Given an SG $(G,L,F)$ and $\lambda \in [0,1]^{|L|}$, \emph{\textsf{o-SCE}($\lambda$)} can be solved with $|L|+2$ queries to an oracle $\Or$. 
	%
	%
\end{restatable}

\begin{corollary}\label{cor:sce_easy}
	Given an SG $(G,L,F)$, if there is a poly-time oracle $\Or$, then \emph{\textsf{o-SCE}($\lambda$)} can be solved in polynomial time.
\end{corollary}

\subsection{Computing SCE-PAs}

First, we provide a positive result: one can find \emph{an} SCE-PA with polynomially many invocations to a stability oracle.
%
It is sufficient to compute $\bx = [x_\pi]$ where $x_{\{p\}}=\Or(G,{c}_p,\varnothing, \varnothing)$ for $p \in L$ and $x_\varnothing=\Or(G,{c}_\lambda, L , \{x_{\{p\}}\}_{p \in L})$.
Thus:

%
%
%

\begin{restatable}{theorem}{theoremseven}
\label{thm:sce_pa_easy}
	Given an SG $(G,L,F)$, \emph{\textsf{f-SCE-PA}} can be solved with $|L|+1$ queries to an oracle $\Or$.
	%
\end{restatable}

\begin{corollary}\label{cor:sce_pa_easy}
	Given an SG $(G,L,F)$, if there is a poly-time oracle $\Or$, then \emph{\textsf{f-SCE-PA}} can be solved in polynomial time.
\end{corollary}

Now, we switch to the problem of computing an optimal SCE-PA, showing that it cannot be solved efficiently, even with access to a polynomial-time stability oracle.
Specifically, we prove a stronger negative result: even the easier problem of verifying the perfect stability of a given $\bx\in \bX$ is computationally intractable.
Our statement is based on a reduction from the \textsf{coNP}-complete problem of deciding whether a given formula in \emph{disjunctive normal form} (DNF) is a tautology or not~\cite{arora2009computational}.
%

%
%
%
%
%

\begin{restatable}{theorem}{theoremeight}
\label{thm:hardness_verify}
	Given an SG $(G,L,F)$ and $\bx \in \bX$, verifying whether $\bx \in$ or $\not \in  \bX^\textsc{PS}$ is not in \emph{\textsf{P}} unless \emph{\textsf{NP} = \textsf{coNP}}, even with access to a polynomial-time decision-form oracle $\Or^\textsc{d}$.
\end{restatable}

\begin{proof}[Proof sketch.] 
	Given a formula $\Phi$ in DNF, we construct an SG $(G,L,F)$ involving a leader $p_v$ for each variable $v \in V$ and a single follower $p_f$.
	Each $p_v$ has two actions, $s_\textsc{T}$ and $s_\textsc{F}$, which define the truth value of $v$. As a result, any $s \in S$ corresponds to a truth assignment $\tau^s$ defined by leaders' strategies.
	The follower has a strategy $s_v$ for each variable $v \in V$.
	Table~\ref{table:reduction} reports the leaders' utilities (the follower always gets $0$).
	We build $\bx = [x_\pi] \in \bX$ with $x_\varnothing(s) =1$ for some $s \in S$ such that $s_{p_v} = s_\textsc{T}$ for all $v \in V$.
	%
	Furthermore, for every $v \in V$ and $\pi \in \Pi_{L \setminus \{p_v\}}$, we let $x_{\pi p_v}(s) = 1$ for $s \in S$ with $s_{p}= s_\textsc{F}$ for all $p \in \pi p_v$, $s_p=s_\textsc{T}$ for all $p \in L\setminus \pi p_v$, and $s_{p_f}=s_{v}$.

	\emph{If.}
	%
	Suppose $\Phi$ is a tautology.
	%
	For every $\pi \in \Pi_L$, $x_\pi$ recommends all the leaders in $\pi$ to play $s_\textsc{F}$.
	%
	%
	%
	Note that, for every $v \in V$ and $\pi \in \Pi_{L \setminus \{p_v\}}$, $u_{p_v}(x_\pi)= \#\textsc{F}(\tau^s) =|\pi|$, while, if $p_v$ decides to \textsc{Opt-Out}, she is recommended to play $s_\textsc{F}$ (one more variable is set to false) and, being $s_{p_f}=s_v$, she gets the same utility. 
	As a result, all distributions $x_\pi$ are stable.

	\emph{Only if.}
	Suppose $\Phi$ is not a tautology.
	%
	Let $s \in S$ be such that $\Phi(\tau^s) = \textsc{F}$.
	If $s_{p_v} =s_\textsc{T}$ for every $v \in V$, then $x_\varnothing$ is not stable as the leaders would \textsc{Opt-Out} (getting at least $0 > -1$).
	Otherwise, there exists $s' \in S$ such that $\Phi(\tau^{s'}) = \textsc{T}$ and $x_{\pi}(s')=1$, $x_{\pi p_v}(s)=1$ for some $v \in V$, $\pi\in \Pi_{L \setminus \{p_v\}}$.
	Then, $u_{p_v}(x_\pi)= \#\textsc{F}(\tau^{s'}) = |\pi|$ and $u_{p_v}(x_{\pi p_v})= |V| > |\pi|$.
	Thus, $x_\pi$ is not stable, as leader $p_v$ would \textsc{Opt-Out}.
	%
	%
	%
	%
	%
	%
%
%
%
%
\end{proof}
\begin{table}[t]
	\centering
	{\renewcommand{\arraystretch}{1.1}\setlength{\tabcolsep}{3pt}\begin{tabular}{c|c|c|c|c|}	
			\multicolumn{1}{c}{}& \multicolumn{2}{c}{$\Phi(\tau^s)=\textsc{T}$} & \multicolumn{2}{c}{$\Phi(\tau^s)=\textsc{F}$}\\
			\multicolumn{1}{c}{} & \multicolumn{1}{c}{ $s_{p_f}=s_v$}  & \multicolumn{1}{c}{$s_{p_f} \neq s_v$} & \multicolumn{1}{c}{$\exists v : s_{p_v} = s_\textsc{F} $} &  
			\multicolumn{1}{c}{$\forall v : s_{p_v}=s_\textsc{T}$}  \\ \cline{2-5}
			$s_\textsc{T}$ & $0$ & $\#\textsc{F}(\tau^s)$ &$|V|$& $-1$ \\ \cline{2-5}
			$s_\textsc{F}$ & $\#\textsc{F}(\tau^s)-1$ & $|V|$ &$|V|$& $0$\\\cline{2-5}
	\end{tabular}}
	\caption{Leader $p_v$'s ($v \in V$) utilities in the SG for the reduction of Theorem~\ref{thm:hardness_verify}. On rows, there are $p_v$'s strategies $s_\textsc{T}$ and $s_\textsc{F}$, whereas, on columns, we report the four possible cases for $s \in S$. $\#\textsc{F}(\tau^s)$ denotes the number of variables set to false by $\tau^s$.}
	\label{table:reduction}
\end{table}

As a byproduct of Theorem~\ref{thm:sce_easy} we have that, when looking for optimal SCEs, one can restrict the attention to those $\bx \in \bX$ admitting a representation whose size is polynomial in the size of the game.
For Theorem~\ref{thm:sce_pa_easy}, the same holds when searching for an SCE-PA.
However, Theorem~\ref{thm:hardness_verify} implies that optimal SCE-PAs require an exponential number of different distributions.
Moreover, even when $\bx \in \bX$ can be easily represented in a compact form (as in the proof of Theorem~\ref{thm:hardness_verify}), we cannot check in polynomial time whether $\bx \in \bX^\textsc{PS}$ or not.
%

This poses a new intriguing question: can we restrict the attention to $\bx \in \bX$ whose size is less than factorial in the number of players?
%
%
We show that the answer is positive.
%
It is sufficient to consider $\bx\in \bX$ whose size is exponential in the number of players, 
as only the unordered set of defecting leaders and the last of them who decided to \textsc{Opt-Out} matter.

%
%
%
%

\begin{restatable}{theorem}{theoremnine}
\label{thm:fact_to_exp}
	Given an SG $(G,L,F)$ and $\bx=[x_\pi]\in \bX^\textsc{PS}$, there is an $\bx'=[x'_\pi]\in \bX^\textsc{PS}$ s.t. $x'_\varnothing = x_\varnothing$	and $x'_{\pi p}=x'_{\pi' p}$ for every $p\in L$ and $\pi, \pi' \in \Pi_{L\setminus \{p\}}$ defining the same set.
\end{restatable}

%
%
%

%
%
%
%

Theorem~\ref{thm:fact_to_exp} allows us to reduce the number of queries to a stability oracle that are necessary to find an optimal SCE-PA.

\begin{restatable}{theorem}{theoremten}
\label{thm:opt_sce_pa}
	Given an SG $(G,L,F)$ and $\lambda \in [0,1]^{|L|}$, \emph{\textsf{o-SCE-PA($\lambda$)}} can be solved with $|L| 2^{|L| - 1} + 1$ queries to $\Or$.
%
\end{restatable}

%
%
%
%
%
%

Finally, we can provide an example showing that Theorem~\ref{thm:opt_sce_pa} is tight, which leads to the following proposition.

\begin{restatable}{prop}{propositionsix}
	Solving \emph{\textsf{o-SCE-PA($\lambda$)}} requires to take into account the last player who performed \textsc{Opt-Out}, while focusing only on the set of defecting leaders is not sufficient.
\end{restatable}

\footnotetext[9]{We remark that, for normal-form games, a polynomial-time stability oracle $\Or$ can be implemented by using a variation of the LP for finding optimal CEs~\cite{shoham2008multiagent}.}

\section{Stability Oracle for Compact Games}\label{sec:games_main}

We study which classes of games admit a polynomial-time stability oracle $\Or$, focusing on those with polynomial type.~\footnotemark 
%
%
%

\noindent
In this section, we only provide our main final result; a detailed description of all the ancillary results is in Appendix~\ref{sec:games}.

Inspired by the classical approaches for finding CEs in games with polynomial type \cite{papadimitriou2008computing,leyton2011ellipsoid,jiang2015polynomial}, we solve $\Or(G,c,L,\{x_p  \}_{p\in L' \subseteq L})$ in polynomial time using the ellipsoid method.
This requires that a suitably defined separation problem (\textsf{Sep}$(z,t)$) can be computed in polynomial time.
Our main result is that \textsf{Sep}$(z,t)$ can be reduced to the \emph{weighted deviation-adjusted social welfare problem} (\textsf{w-DaSW}$(y,v,t)$) introduced by~\citeauthor{leyton2011ellipsoid}~[\citeyear{leyton2011ellipsoid}] for finding an optimal (according to some linear function of players' utilities) CE.
This establishes a strict connection between the problem solved by our stability oracle and that of computing optimal CEs. 
As a consequence, given the results of~\citeauthor{leyton2011ellipsoid}~[\citeyear{leyton2011ellipsoid}], $\Or$ can be computed in polynomial time for all the compact games where finding an optimal CE is computationally tractable.
Thus:

\begin{theorem}\label{thm:compact_games}
	The following games admit a polynomial-time stability oracle $\Or$: anonymous games, symmetric games, and bounded-treewidth graphical and polymatrix games.
\end{theorem}

Finally, our results also imply that the polynomial-time stability oracle $\Or$ always outputs a polynomially-sized correlated distribution (see Corollary~\ref{cor:poly_support} in Appendix~\ref{sec:games}).

\section{Discussion}\label{sec:discussion}

This paper introduces a new way to apply the Stackelberg paradigm to any (underlying) finite game.
Differently from previous works, our approach deals with scenarios involving multiple leaders by introducing a preliminary agreement stage in which each leader can decide whether to be a leader or become a follower.
We introduce and study three natural solution concepts that differ depending on the properties that they require on the agreement stage (others, \emph{e.g.}, requiring stability and perfect efficiency, will be explored in future). 

Our equilibria generalize the optimal correlated strategies to commit to introduced by~\citeauthor{conitzer2011commitment}~[\citeyear{conitzer2011commitment}] for single-leader multi-follower Stackelberg games.
At the same time, they also provide a significant advancement over the multi-leader solution concepts introduced in the security context (see, \emph{e.g.}, \cite{gan2018stackelberg}).
First, correlated-strategy commitments are more natural than leaders' strategies satisfying some Nash-like constraints.
Second, our equilibria are funded on strong game-theoretic groundings, as they are guaranteed to exist independently of the game structure.
Last but not least our solutions apply to general games.

Finally, our computational findings exploit a general framework relying on a game-independent stability oracle.
%
%
Thus, our positive results can be extended to other game classes by simply designing polynomial-time oracles.

%

\section*{Acknowledgments}

This work has been partially supported by the Italian MIUR PRIN 2017 Project ALGADIMAR ``Algorithms, Games, and Digital Market''.

\clearpage

\bibliographystyle{named}
\bibliography{refs}

\clearpage

\appendix

\section{SCEs and non-Stackelberg Correlation}\label{sec:sce_and_cor}

We analyze how our solution concepts relate to other non-Stackelberg solutions involving correlation. 
Specifically, we focus on CEs (see Section~\ref{sec:prelim} for a formal definition) and their \emph{coarse} variant, which we define in the following.

The {coarse} CE weakens the CE by only enforcing protection against \emph{a priori} defections, \emph{i.e.}, before the recommendations are revealed to the players~\cite{moulin1978}.
Formally, $x \in \X$ is a \emph{coarse correlated equilibrium} (CCE) if, for every player $p \in N$ and strategy $s_p' \in S_p$, the following constraint holds: 
\begin{equation}\label{eq:incentive_cce}
\sum_{s \in S} x(s) \left( u_p(s) - u_p(s_p',s_{-p})  \right) \geq 0.
\end{equation}
We denote with $\X^\textsc{CCE}$ the set of CCEs of the game.

In our analysis, we compare CEs and CCEs with the correlated distributions $x_\varnothing$ resulting from our solution concepts in general SGs.
Given an SG $(G,L,F)$, we define $\X^\textsc{S} \subseteq \X$ and $\X^\textsc{PS} \subseteq \X$ as the sets of $x_\varnothing$ such that $\bx =[x_\pi] \in \bX^\textsc{S}$ and $\bx \in \bX^\textsc{PS}$, respectively.
Our goal is to investigate the relationships involving the sets $\X^\textsc{S}$ and $\X^\textsc{PS}$ with the sets of CEs and CCEs of the underlying game $G$, namely $\X^\textsc{CE}$ and $\X^\textsc{CCE}$. 
%
%
Figure~\ref{fig:diagram} depicts these relationships.

\begin{figure}[H]
	\centering
	\includegraphics[width=0.27\textwidth]{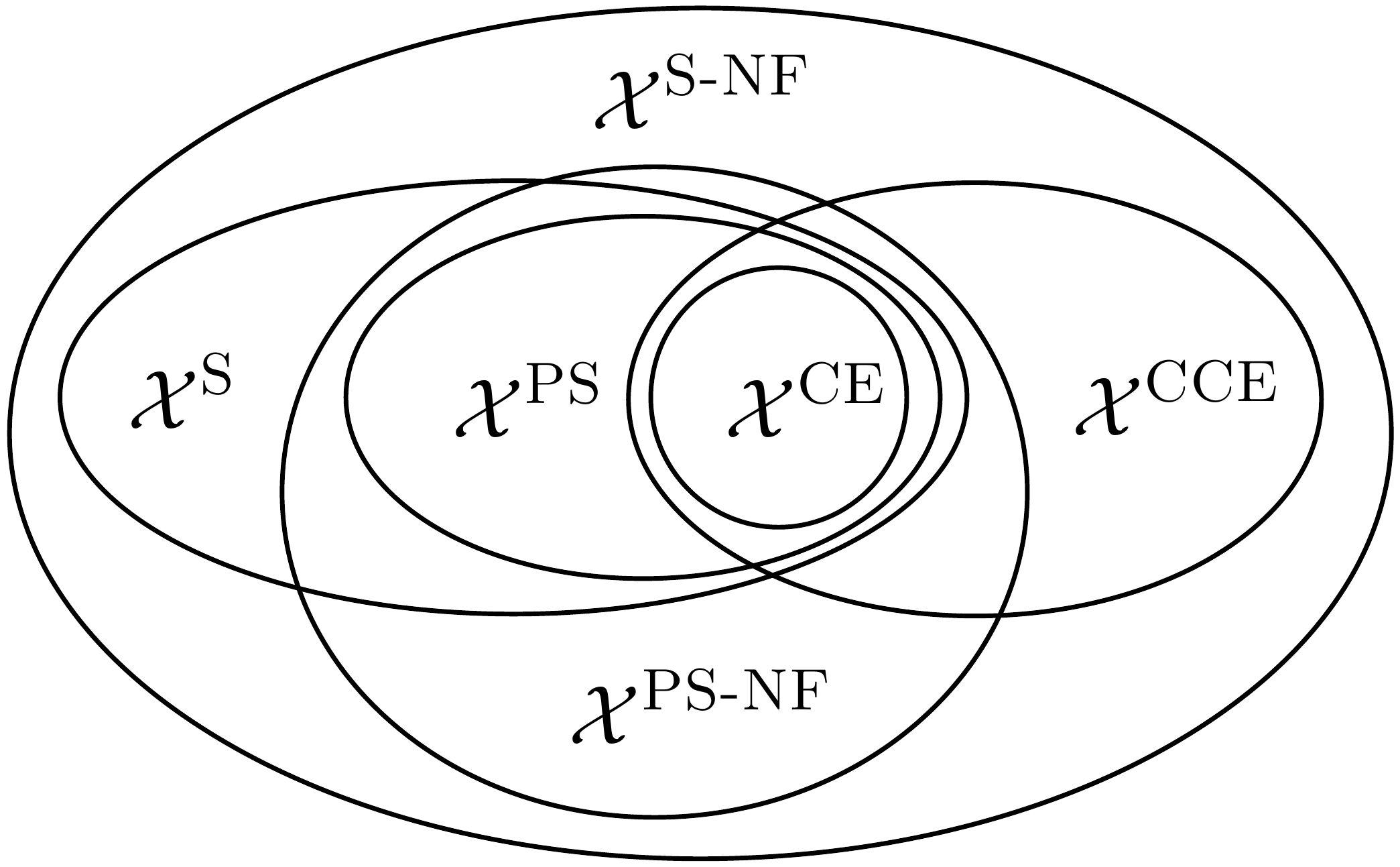}
	\caption{Relations among $\X^{\textsc{S}}$, $\X^{\textsc{PS}}$, $\X^{\textsc{CE}}$, $\X^{\textsc{CCE}}$, $\X^{\textsc{S-NF}}$, $\X^{\textsc{PS-NF}}$.}
	\label{fig:diagram}
\end{figure}

Let us remark that the relations $\X^\textsc{CE} \subseteq \X^\textsc{CCE}$, $\X^\textsc{CE} \subseteq \X^\textsc{PS}$, and $\X^\textsc{PS} \subseteq \X^\textsc{S}$ hold by definition, while it is easy to show that $\X^\textsc{CE} \subseteq \X^\textsc{PS}$ (see the proof of Theorem~\ref{thm:existence_sce}).

First, we look at the connection between (perfectly) stable distributions and CCEs.
Given the relation between SEs and SCEs (see Corollary~\ref{cor:corollary_se}) in single-leader single-follower SGs, the following result holds as a direct consequence of~\cite[Remark~13]{von2010leadership}.

\begin{restatable}{prop}{propositiontwo}
	There are SGs where $\X^\textsc{CCE} \nsubseteq  \X^\textsc{S} $.
\end{restatable}

Moreover, not all perfectly stable distributions are CCEs.

\begin{restatable}{prop}{propositionthree}
	\label{obs:PS-CCE}
	There are SGs where $\X^\textsc{PS} \nsubseteq \X^\textsc{CCE}$.
\end{restatable}

\begin{proof}
	Consider the SG in Table~\ref{table:games_for_existence}~(Left), where $L = \{1,2\}$ and $F = \varnothing$. 
	Since $s_{1,1}$ and $s_{2,1}$ are strictly dominated, there is a unique CCE $x \in \X^\textsc{CCE}$ with $x(s_{1,2},s_{2,2})=1$.
	Let $\bx = [x_\pi] \in \bX$ be such that $x_\varnothing(s_{1,1},s_{2,1})=1$ and $x_\pi(s_{1,2},s_{2,2}) = 1$ for all $\pi \neq \varnothing \in \Pi_L  $.
	Notice that each $x_\pi$ with $\pi \neq \varnothing$ satisfies the incentive constraints of Eq.~\eqref{eq:incentive_ce} for every player, and, thus, $x_\pi \in \X^\textsc{CE}_{\pi }$.
	Moreover, for each leader $p \in L$, $u_p(x_\varnothing) = 2$ and $u_p(x_\pi)=1$ for all $\pi \in \Pi_L\setminus\{\varnothing\}$.
	Thus, each $x_\pi$ is stable and $\bx \in \bX^\textsc{PS}$.
	%
\end{proof}
\begin{table}[!htp]
	\centering
	\begin{minipage}{0.18\textwidth}
		{\renewcommand{\arraystretch}{1.1}\setlength{\tabcolsep}{2pt}\begin{tabular}{cc|c|c|}	
				& \multicolumn{1}{c}{} & \multicolumn{1}{c}{$s_{2,1}$}  & \multicolumn{1}{c}{$s_{2,2}$} \\\cline{3-4}
				& $s_{1,1}$ & $2,2$ & $0,3$ \\\cline{3-4}
				& $s_{1,2}$ & $3,0$ & $1,1$ \\\cline{3-4}
		\end{tabular}}
	\end{minipage}
	\begin{minipage}{0.29\textwidth}
		{\renewcommand{\arraystretch}{1.1}\setlength{\tabcolsep}{2pt}\begin{tabular}{c|c|c|c|c|}	
				\multicolumn{1}{c}{}&\multicolumn{1}{c}{$s_{2,1}$}  & \multicolumn{1}{c}{$s_{2,2}$} & \multicolumn{1}{c}{$s_{2,3}$} & \multicolumn{1}{c}{$s_{2,4}$} \\\cline{2-5}
				$s_{1,1}$ & $0,0$ & $-2,4$ &$1,-8$&$1,-2$\\\cline{2-5}
				$s_{1,2}$ & $1,-8$ & $0,0$ &$-2,4$&$1,-2$\\\cline{2-5}
				$s_{1,3}$ & $-2,4$ & $1,-8$ &$0,0$&$1,-2$\\\cline{2-5}
		\end{tabular}}
	\end{minipage}
	\caption{\emph{Left}: Two-player normal-form SG where $\X^\textsc{PS} \nsubseteq  \X^\textsc{CCE} $. \emph{Right:} Two-player normal-form SG where $\X^\textsc{CCE} \nsubseteq \X^\textsc{PS-NF}$.}
	\label{table:games_for_existence}
\end{table}

Next, we analyze the relationships with the sets $\X^\textsc{S-NF}$ and $\X^\textsc{PS-NF}$, which are defined as $\X^\textsc{S}$ and $\X^\textsc{PS}$, but for the SG $(G,N,\varnothing)$ where each player is a leader.
Our goal is to study the impact of players' roles in SGs having the same underlying finite game.
The following result shows that enlarging the set of leaders can only introduce new stable distributions. 

\begin{restatable}{theorem}{theoremfour}
	$\X^\textsc{S} \subseteq \X^\textsc{S-NF}$ and $\X^\textsc{PS} \subseteq \X^\textsc{PS-NF}$.
\end{restatable}
\begin{proof}
	We only prove the result for $\X^\textsc{PS}$, as similar arguments hold for $\X^\textsc{S}$.
	Given any SG $(G,L,F)$, for every perfectly stable $\bx =[x_\pi] \in\bX^\textsc{PS}$ of $(G,L,F)$, we show that there exists a perfectly stable $\bx' = [x'_\pi] \in \bX^\textsc{PS}$ of $(G,N,\varnothing)$, such that $x_\varnothing=x'_\varnothing$.
	Let us define $x'_\pi = x_{\pi \cap L}$, for all $\pi \in \Pi_N$.
	Clearly, it holds $x'_\pi \in \X^\textsc{CE}_\pi$, as $x'_\pi \in \X^\textsc{CE}_{(\pi \cap L)\cup F}  \subseteq \X^\textsc{CE}_{\pi}$.
	For every player $p \in L$ and $\pi \in \Pi_N$ such that $p \notin \pi$, we have $u_p(x'_\pi) = u_p(x_{\pi \cap L})$ and $u_p(x'_{\pi p}) = u_p(x_{\pi p \cap L})$.
	Thus, given that $\bx \in \bX^\textsc{PS}$, $\bx'$ satisfies the stability constraints for the players in $L$.
	Now, in order to show that $\bx' \in \bX^\textsc{PS}$, it is sufficient to prove that players in $F$ do not have an incentive to \textsc{Opt-Out} in $(G,N,\varnothing)$.
	This is the case as, for $p \in F$ and $\pi \in \Pi_N$ with $p \notin \pi$, we have $x'_{\pi p} = x'_{\pi}$.
	%
	%
	%
	%
	%
\end{proof}

Furthermore, we can also provide examples showing that:

\begin{restatable}{prop}{propositionfour}
	There are SGs where $\X^\textsc{PS-NF} \nsubseteq \X^\textsc{S}$.
\end{restatable}

\begin{proof}
	Consider the SG in Table~\ref{table:games_for_existence}~(Left), where $L = \{1\}$ and $F =  \{2\}$.
	There is an $\bx=[x_\pi] \in \bX^\textsc{PS-NF}$ of $(G,N,\varnothing)$ in which $x_\varnothing(s_{1,1},s_{2,1})=1$ (see the proof of Proposition~\ref{obs:PS-CCE}). 
	Let $\bx'=[x'_\pi] \in \bX^\textsc{S}$ of $(G,L,F)$. Since $x'_\varnothing \in \mathcal{X}^\textsc{CE}_{\{2\}}$ and $s_{2,1}$ is strictly dominated, it must be $x'_\varnothing(s_{1,1},s_{2,1})=0$.
\end{proof}

\begin{restatable}{prop}{propositionfive}
	There are SGs where $\X^\textsc{CCE} \nsubseteq  \X^\textsc{PS-NF} $.
\end{restatable}

\begin{proof}
	Consider the SG in Table~\ref{table:games_for_existence}~(Right), where $L = N =  \{1,2\}$.
	There is a CCE $x \in \X^\textsc{CCE}$ with $x(s_{1,1},s_{2,1})=x(s_{1,2},s_{2,2})=x(s_{1,3},s_{2,3})=\frac{1}{3}$.
	We show that there is no $\bx = [x_\pi] \in \bX^\textsc{PS}$ with $x_\varnothing=x$.
	By contradiction, assume there exists such $\bx$.
	Given that $u_1(x_\varnothing)=0$, it should be the case that $u_1(x_{\{1\}}) \leq 0$, by stability of $x_\varnothing$.
	%
	%
	Take the incentive constraints of player 1 (Eq.~\eqref{eq:incentive_ce}).
	%
	%
	Since there must be no incentive to deviate from $s_{1,1}$ to $s_{1,2}$, it holds 
	%
	$$
	x_{\{1\}}(s_{1,1},s_{2,3}) \geq \frac{1}{3}x_{\{1\}}(s_{1,1},s_{2,1}) +\frac{2}{3} x_{\{1\}}(s_{1,1},s_{2,2}).
	$$ 
	Similar conditions also hold for the deviation from $s_{1,2}$ to $s_{1,3}$ and that from $s_{1,3}$ to $s_{1,1}$.
	Thus, we can write:
	\begin{align*}
	x_{\{1\}}(s_{1,2},s_{2,1}) \geq \frac{1}{3}x_{\{1\}}(s_{1,2},s_{2,2}) +\frac{2}{3} x_{\{1\}}(s_{1,2},s_{2,3}), \\
	x_{\{1\}}(s_{1,3},s_{2,2}) \geq \frac{1}{3}x_{\{1\}}(s_{1,3},s_{2,3}) +\frac{2}{3} x_{\{1\}}(s_{1,3},s_{2,1}).
	\end{align*}
	As a result, we can conclude that, if $x_{\{1\}}$ only recommends player 2 to play $s_{2,1}$, $s_{2,2}$, and $s_{2,3}$, then $u_2(x_{\{1\}}) < -2$.
	However, if player 2 decides to \textsc{Opt-Out}, then she would get at least $-2$, as $x_{\{1,2\}} \in \X^\textsc{CE}$ and player 2 is guaranteed to get $-2$ by playing $s_{2,4}$.
	Thus, being $x_{\{1\}}$ stable, it must be the case that player 2 is always recommended to play $s_{2,4}$ in $x_{\{1\}}$.
	Thus, $u_1(x_{\{1\}})=1$, which is a contradiction.
	%
	%
	%
\end{proof}

Finally, we prove that the stable distributions for the SG without followers encompass those defining CCEs.

\begin{restatable}{theorem}{theoremfive}
	$\X^\textsc{CCE} \subseteq \X^\textsc{S-NF}$.
\end{restatable}

\begin{proof}
	Let $x \in \X^\textsc{CCE}$ be a CCE of a given finite game $G$.
	We prove that the SG $(G,N,\varnothing)$ admits a stable distribution $\bx = [x_\pi] \in \bX^\textsc{S}$ with $x_\varnothing= x$.
	In order to do so, for every leader $p \in N$, we let $x_{\{p\}}$ be such that $u_p(x_{\{p\}}) \leq u_p(x)$, as shown in the following.
	Let us fix a player $p \in N$ and let $\hat s_p \in S_p$ be such that, for every $s_p' \in S_p$:
	\begin{equation}\label{eq:proof_CCE_SNF_eq1}
	\sum_{s \in S} x(s) \left( u_p(\hat s_p,s_{-p}) - u_p(s_p',s_{-p}) \right) \geq 0,
	\end{equation}
	\emph{i.e.}, $\hat s_p$ is the best player $p$'s strategy against the correlated distribution $x$.
	Given that $x \in \X^\textsc{CCE}$:
	\begin{equation}\label{eq:proof_CCE_SNF_eq2}
	\sum_{s \in S} x(s) \left( u_p(s) - u_p(\hat s_p,s_{-p})  \right) \geq 0.
	\end{equation}
	We define $x_{\{p\}}$ as follows:
	\begin{itemize}
		\item $x_{\{p\}}(\hat s_p,s_{-p})=\sum_{s_p \in S_p} x(s_p,s_{-p}) \,\, \forall s_{-p} \in S_{-p}$;
		\item $x_{\{p\}}(s_p,s_{-p})=0 \,\, \forall s_p \neq \hat s_p \in S_p, s_{-p}\in S_{-p}$. 
	\end{itemize}
	Given how $x_{\{p\}}$ is defined and Eq.~\eqref{eq:proof_CCE_SNF_eq1}, we have that the incentive constraints of player $p$ (Eq.~\eqref{eq:incentive_ce}) are satisfied, and, thus, $x_{\{p\}} \in \X^\textsc{CE}_{\{p\}}$.
	Moreover, Eq.~\eqref{eq:proof_CCE_SNF_eq2} implies that $u_p(x_{\{p\}}) \leq u_p(x)$, which concludes the proof.
	%
	%
	%
	%
	%
	%
	%
\end{proof}

Observe that, when one looks for equilibria maximizing a linear function of leaders' utilities (\emph{e.g.}, the leaders' social welfare), larger sets result in better solutions.~\footnote{Let us remark that, since $\bX^\textsc{S}$ and $\bX^\textsc{PS}$ are polytopes (see Lemma~\ref{lem:polytope} in Appendix~\ref{sec:props_existence_app}), maximizing a linear function of leaders' utilities over the sets $\bX^\textsc{S}$ and $\bX^\textsc{PS}$ also provides Pareto optimality, and, thus, efficiency over the corresponding set.}
Moreover, we can provide examples where the difference in terms of leaders' social welfare between two solution concepts can be arbitrarily large.
For instance, the following holds.~\footnote{Similar results hold for the other pairs of solution concepts.}

\begin{restatable}{prop}{propositionSW}
	There are SGs $(G,L,F)$ with leaders' social welfare in SCE-PAs arbitrarily larger than in any CE of $G$.
\end{restatable}

\begin{proof}
	Consider the SG in Table~\ref{table:MLbetter}, where $L = \{1,2\}$ and $F=\varnothing$. 
	Since strategies $s_{1,1}$ and $s_{2,1}$ are strictly dominated, the only CE is $x \in \X^\textsc{CE}$ with $x(s_{1,2},s_{2,2})=1$.
	Let $\bx = [x_\pi] \in \bX$ be such that $x_\varnothing(s_{1,1},s_{2,1})=1$ and $x_\pi(s_{1,2},s_{2,2}) = 1$ for all $\pi \neq \varnothing \in \Pi_L $.
	It is easy to check that $\bx \in \bX^\textsc{SCE-PA}$.
	Moreover, the social welfare of the CE is $2$, while the social welfare of the SCE-PA is $2k$. 
\end{proof}	

\begin{table}[!htp]
	\centering
	{\renewcommand{\arraystretch}{1.1}\setlength{\tabcolsep}{2pt}\begin{tabular}{cc|c|c|}	
			& \multicolumn{1}{c}{} & \multicolumn{1}{c}{$s_{2,1}$}  & \multicolumn{1}{c}{$s_{2,2}$} \\\cline{3-4}
			& $s_{1,1}$ & $k,k$ & $0,k+1$ \\\cline{3-4}
			& $s_{1,2}$ & $k+1,0$ & $1,1$ \\\cline{3-4}
	\end{tabular}}
	\caption{Two-player normal-form SG (with $k > 0$) where the leaders' social welfare of an SCE-PA is arbitrary larger than in any CE.}
	\label{table:MLbetter}
\end{table}

\section{Omitted Proofs for Section~\ref{sec:props_existence}}\label{sec:props_existence_app}



\begin{lemma}\label{lem:polytope}
	The sets $\bX^\textsc{S}$ and $\bX^\textsc{PS}$ are polytopes.
\end{lemma}

\begin{proof}
	$\bX \subseteq \mathbb{R}^{|\Pi_L| \cdot |S|}$ is the set of vectors $\bx = [x_\pi]$ such that $x_\pi \in \X^\textsc{CE}_{\pi \cup F}$ for all $\pi \in \Pi_L$. 
	Each $\X^\textsc{CE}_{\pi \cup F}$ is defined by the linear constraints of Eq.~\eqref{eq:incentive_ce}, thus $\bX$ is a polytope.
	Moreover, if $\bx \in \bX^\textsc{PS} \subseteq \bX$, $x_\pi$ is stable for all $\pi \in \Pi_L$, \emph{i.e.}, $u_p(x_\pi) \geq u_p(x_{\pi p})$ for all $p \in L \setminus \pi$.
	Thus, being these constraints linear, $\bX^\textsc{PS}$ is a polytope.
	A similar argument holds for $\bX^\textsc{S}$.
	%
	%
	%
	%
	%
\end{proof}

\theoremtwo*

\begin{proof}
	Given an SG $(G,L,F)$, let $x \in \X^\textsc{CE}$ and $\bx =[x_\pi] \in \bX$ be such that $x_\pi = x$ for all $\pi \in \Pi_L$.
	We prove that $\bx \in \bX^\textsc{PS}$.
	First, for each $\pi \in \Pi_L$, $x_\pi \in \X^\textsc{CE}_{\pi \cup F}$, since $ \X^\textsc{CE} \subseteq \X^\textsc{CE}_{\pi \cup F}$.
	Moreover, each $x_\pi$ is stable, since $u_p(x_\pi) = u_p(x_{\pi p})$ for all $p \in L \setminus \pi$.
	This shows that $\bX^\textsc{PS} \neq \varnothing$.
	Finally, being $\bX^\textsc{PS}$ a polytope by Lemma~\ref{lem:polytope}, there exists $\bx = [x_\pi] \in \bX^\textsc{PS}$ such that $x_\varnothing \in \mathcal{P}_L(\bX^\textsc{PS})$.
	%
	Thus, $\bX^\textsc{SCE-PA} \neq \varnothing$.
	A similar reasoning holds for the sets $\bX^\textsc{S}$ and $\bX^\textsc{SCE}$.
	%
%
\end{proof}

\propositionone*
\begin{proof}
	Consider the SG in Table~\ref{table:no_sce_pape},
	%
	where $L=\{1,2,3\}$ and $F=\varnothing$.
	Suppose, by contradiction, that there exists $\bx = [x_\pi] \in \bX^\textsc{SCE-PAPE}$. 
	%
	%
	First, for every $x_\pi$ with player 3 in $\pi$, $u_3(x_\pi) = 1$ (otherwise $x_\pi \notin \X^\textsc{CE}_{\pi \cup F}$, as player 3 always gets $1$ by deviating to $s_{3,3}$). 
	Let us consider the sequences of \textsc{Opt-Out} defined by the ordered subsets ${\{1,2\}}$ and ${\{2,1\}}$.
	Given that the definition of stability requires $u_3(x_{\{1,2\}}) \geq u_3(x_{\{1,2,3\}}) = 1$ and $u_3(x_{\{2,1\}}) \geq u_3(x_{\{2,1,3\}}) = 1$, we have that $x_{\{1,2\}}$ and $x_{\{2,1\}}$ must place strictly positive probability only on strategy profiles $(s_{1,2},s_{2,2},s_{3,1})$, $(s_{1,1},s_{2,2},s_{3,2})$, and those recommending $s_{3,3}$ to player 3.
	Moreover, player 1 cannot be told to play $s_{1,2}$, as she would have an incentive to deviate to $s_{1,1}$.
	The same holds for player 2 and strategy $s_{2,2}$.
	As a result, $x_{\{1,2\}}$ and $x_{\{2,1\}}$ must always recommend $s_{3,3}$ to player 3.
	%
	%
	%
	Now, let us take the sequence of \textsc{Opt-Out} defined by $\{1\}$.
	By stability of $x_{\{1\}}$, it must hold $u_3(x_{\{1\}}) \geq u_3(x_{\{1,3\}}) = 1$ and $u_2(x_{\{1\}}) \geq u_2(x_{\{1,2\}}) = 0$.
	Hence, given $x_{\{1\}} \in \X^\textsc{CE}_{\{1\}}$, we can conclude that, in order to satisfy $x_{\{1\}} \in \mathcal{P}_{L \setminus \{1\}}(\bX^\textsc{PS})$, $x_{\{1\}}$ must always recommend the strategy profile $(s_{1,1},s_{2,2},s_{3,2})$, where player 1 gets a utility of $2$.
	Similarly, for the sequence defined by $\{2\}$, $x_{\{2\}}$ must always recommend $(s_{1,2},s_{2,2},s_{3,1})$ and, thus, player 2 receives a utility of $2$.
	Thus, for stability, $x_\varnothing$ must satisfy $u_1(x_\varnothing), u_2(x_\varnothing) \geq 2$, which is clearly impossible.
	%
%
%
%
\end{proof}

\section{Omitted Proofs for Section~\ref{sec:prop_relations}}\label{sec:prop_relations_app}


\theoremthree*

\begin{proof}
	Since the SG has only one leader (player 1), stability and perfect stability are equivalent, and, thus, $\bX^\textsc{S} = \bX^\textsc{PS}$.
	As a result, $\bX^\textsc{SCE}=\bX^\textsc{SCE-PA}$.
	Moreover, for the same reasons, also efficiency and perfect efficiency are equivalent, and $\bX^\textsc{SCE-PA}=\bX^\textsc{SCE-PAPE}$.
	Note that requiring Pareto optimality is the same as maximizing the leader's utility function $u_1$.
	Let $\bx = [x_\pi] \in \bX^\textsc{SCE}$ and assume, by contradiction, that $x_\varnothing$ is not an optimal correlated strategy to commit to.
	This would imply that there exists another $\hat x \in \X^\textsc{CE}_{N \setminus \{1\}}$ such that $u_1(\hat x) \geq u_1(x_\varnothing)$.
	However, replacing $x_\varnothing$ with $\hat x$ in $\bx$ would give us another $\hat \bx \in \bX^\textsc{S}$ (stability constraints are trivially satisfied).
	This would contradict the efficiency of $\bx$.
\end{proof}

\section{Omitted Proofs for Section~\ref{sec:computational}}\label{sec:computational_app}

\theoremsix*
\begin{proof}
	We build an $\bx = [x_\pi] \in \bX^\textsc{SCE}$ that maximizes $f_\lambda$ by invoking a stability oracle $\Or$ multiple times.
	For every $p \in L$, we define $x_{\{p\}}=\Or(G,{c}_p,L \setminus \{p\},\varnothing)$.
	Moreover, we let $x_\varnothing=\Or(G,{c}_\lambda,L, \{x_{\{p\}}\}_{p\in L})$ and $x_\pi=\Or(G,{c}_\lambda,\varnothing, \varnothing)$ for every $\pi\in \Pi_L$ with $|\pi| \geq 2$.  
	%
	Clearly, we need $|L|+2$ calls to $\Or$.
	First, $x_\pi \in \X^\textsc{CE}_{\pi \cup F}$ for every $\pi \in \Pi_L$, by definition of $\Or$.
	%
	%
	For the same reason, we have $u_p(x_\varnothing) \geq u_p(x_{\{p\}})$ for all $p \in L$.
	Thus, we can conclude that $\bx \in \bX^\textsc{S}$.
	Let $f_\lambda$ be the value of the objective for $\bx$.	
	We show that $f_\lambda$ is maximized over $\bX^\textsc{S}$.
	Being $f_\lambda$ a linear combination of leader's utility functions, we immediately get that $x_\varnothing \in \mathcal{P}_{L}(\bX^\textsc{S})$, and $\bx \in \bX^\textsc{SCE}$.
	By contradiction, suppose that there exists an $\bx' = [x'_\pi] \in \bX^\textsc{S}$ with objective function value $f'_\lambda > f_\lambda$.
	This implies that there exists a leader $p \in L$ with $u_p(x'_{\{p\}})<u_p(x_{\{p\}})$, otherwise the solution $x_\varnothing$ returned by $\Or(G,{c}_\lambda,L, \{x_{\{p\}}\}_{p\in L})$ would not be optimal.
	This is a contradiction, since $x_{\{p\}}$ minimizes player $p$'s utility on the set $\X_{\{p\}\cup F}^\textsc{CE}$, and $x'_{\{p\}} \in \X_{\{p\}\cup F}^\textsc{CE}$.
\end{proof}

\theoremseven*
\begin{proof}
	Using $\Or$, we construct an $\bx \in \bX^\textsc{SCE-PA}$.
	Let $x_{\{p\}}=\Or(G,{c}_p,\varnothing, \varnothing)$, \emph{i.e.}, $x_{\{p\}}$ is a CE that minimizes player $p$'s utility.
	Moreover, we define $x_\varnothing=\Or(G,{c}_\lambda, L , \{x_{\{p\}}\}_{p \in L})$ for some $\lambda \in (0,1]^{|L|}$.
	By setting, for every leader $p \in L$, $x_{\pi}=x_{\{p\}}$ for all $\pi \in\Pi_L$ where $p$ is the first to \textsc{Opt-Out}, we have $\bx \in \bX^\textsc{PS}$.
	Clearly, we only require $|L|+1$ queries to $\Or$.
	Now, we prove that $x_\varnothing \in \mathcal{P}_L(\bX^\textsc{PS})$, and, thus, $\bx \in \bX^\textsc{SCE-PA}$.
	By contradiction, suppose that it is not the case, \emph{i.e.}, there exists an $\bx' = [x'_\pi] \in \bX^\textsc{PS}$ with $u_p(x'_\varnothing) \geq u_p(x_\varnothing)$ for all $p \in L$ and $u_q(x'_\varnothing) > u_q(x_\varnothing)$ for some leader $q \in L$.
	By stability of $x_\varnothing$, we have that $u_p(x'_\varnothing) \geq u_p(x_\varnothing) \geq u_p(x_{\{p\}})$ for every $p\in L$. 
	Thus, $x'_\varnothing$ satisfies $u_p(x'_\varnothing) \geq u_p(x_{\{p\}})$ for every leader $p \in L$ (stability), and
	$$
		\sum_{p \in N} \sum_{s \in S} c_{\lambda,p} u_p(s) x'_\varnothing(s) > \sum_{p \in N} \sum_{s \in S} c_{\lambda,p} u_p(s) x_\varnothing(s),
	$$
	which implies that $x'_\varnothing$ verifies the constraints for a solution to $\Or(G,{c}_\lambda,L, \{x_{\{p\}}\}_{p \in L})$, while providing an objective grater than that of $x_\varnothing$.
	This contradicts the correctness of $\Or$.	
	%
%
%
\end{proof}

\theoremeight*
\begin{proof}
	Given a DNF formula $\Phi$, we build an SG and an $\bx \in \bX$ such that $\bx \in \bX^\textsc{PS}$ if and only if $\Phi$ is a tautology. 
	Thus, if one could verify the perfect stability of $\bx$ in polynomial time, then there would be a polynomial-time checkable certificate for the \textsf{coNP}-complete problem of determining whether a DNF formula is a tautology or not~\cite{arora2009computational}.
	This would imply \textsf{NP} = \textsf{coNP}.
	Moreover, given how the SG is built, the result holds even if we get access to a polynomial-time decision oracle $\Or^\textsc{d}$.
%

	\emph{Construction.}
	Given a DNF formula $\Phi$, let $V$ denote the set of variables appearing in $\Phi$.
	We construct an SG $(G,L,F)$ involving a leader for each variable and a single follower, \emph{i.e.}, $L= \{ p_v \mid v \in V \}$ and $F = \{ p_f \}$.
	Moreover, we let $S_{p_f} = \{s_v \mid v \in V \}$ be the set of follower's strategies, one per variable, while the leaders share the strategies $S_{p_v} = \{ s_\textsc{T}, s_\textsc{F} \}$, corresponding to truth values .
	As a result, any strategy profile $s \in S$ corresponds to a truth assignment $\tau^s$ defined by leaders' strategies.
	%
	%
	We write $\Phi(\tau^s)=\textsc{T}$ if $\tau^s$ satisfies $\Phi$, while $\Phi(\tau^s)=\textsc{F}$ otherwise.
	We also denote with $\#\textsc{F}(\tau^s)$ the number of false variables in $\tau^s$.
	%
	%
	Table~\ref{table:reduction} reports the leaders' utilities, while the follower's one is always $0$.
	Then, we build an $\bx = [x_\pi] \in \bX$ with $x_\varnothing(s) =1$ for some $s \in S$ such that $s_{p_v} = s_\textsc{T}$ for every $v \in V$.
	%
	Furthermore, for every $v \in V$ and $\pi \in \Pi_{L \setminus \{p_v\}}$, we let $x_{\pi p_v}(s) = 1$ for $s \in S$ with $s_{p}= s_\textsc{F}$ for every $p \in \pi p_v$, $s_p=s_\textsc{T}$ for every $p \in L\setminus \pi p_v$, and $s_{p_f}=s_{v}$.
	Let us remark that our SG admits a polynomial-time decision oracle $\Or^\textsc{d}(x,L,\{x_p\}_{p \in L' \subseteq L})$, since it can be queried in polynomial time only on polynomially-sized distributions.
	%
	%

	\emph{If.}
	We prove that, if $\Phi$ is a tautology, then $\bx \in \bX^\textsc{PS}$.
	For every $\pi \in \Pi_L$, $x_\pi$ recommends all the leaders in $\pi$ to play $s_\textsc{F}$.
	Moreover, being $\Phi$ a tautology, strategy $s_\textsc{F}$ (weakly) dominates $s_\textsc{T}$ (as it is always the case that $\Phi(\tau^s)=\textsc{T}$).
	Thus, $x_\pi \in \X^\textsc{CE}_{\pi \cup F}$. 
	Note that, for every $v \in V$ and $\pi \in \Pi_{L \setminus \{p_v\}}$, $u_{p_v}(x_\pi)= \#\textsc{F}(\tau^s) =|\pi|$, while, if $p$ decides to \textsc{Opt-Out}, she is recommended to play $s_\textsc{F}$ and, being $s_{p_f}=s_v$, she gets the same utility. 
	As a result, all distributions $x_\pi$ are stable.

	\emph{Only if.}
	We prove that, if $\Phi$ is not a tautology, then $\bx \notin \bX^\textsc{PS}$.
	Let $s \in S$ be such that $\phi(\tau^s) = \textsc{F}$.
	Two cases are possible.
	If $s_{p_v} =s_\textsc{T}$ for every $v \in V$, then $x_\varnothing$ is not stable as the leaders would have an incentive to \textsc{Opt-Out} (since they get at least $0 > -1$).
	If this is not the case, then there exist $s,s' \in S$ such that $\Phi(\tau^s) = \textsc{T}$ and $\Phi(\tau^{s'}) = \textsc{F}$, where $x_{\pi}(s)=1$ and $x_{\pi p_v}(s')=1$ for some $v \in V$ and $\pi\in \Pi_{L \setminus \{p_v\}}$.
	In this case, $u_{p_v}(x_\pi)= \#\textsc{F}(\tau^{s}) = |\pi|$ and $u_{p_v}(x_{\pi p_v})= |V| > |\pi|$.
	Thus, $x_\pi$ is not stable, as leader $p_v$ would have an incentive to \textsc{Opt-Out}.
\end{proof}

\setcounter{theorem}{7}
\begin{corollary}
	\label{corollary:hardness_verify}
	Given an SG $(G,L,F)$ and $\bx \in \bX$, verifying whether $\bx \in \bX$ is an SCE-PA maximizing the social welfare is not in \emph{\textsf{P}} unless \emph{\textsf{NP} = \textsf{coNP}}, even with access to a polynomial-time decision-form oracle $\Or^\textsc{d}$.
\end{corollary}

\begin{proof}
	We can modify the proof of Theorem~\ref{thm:hardness_verify} so that, when $\Phi$ is a tautology, $\bx \in \bX$ is the only perfectly stable distribution maximizing the social welfare.
	In order to do this, it is enough to add a leader with a single action and utility $|V|^2$ if $s_{p_v}=s_\textsc{T}$ for all $v \in V$, while $0$ otherwise.
\end{proof}

\theoremnine*
\begin{proof}
	Let us take some $\bx \in \bX^\textsc{PS}$. 
	For every $p \in L$ and $\pi \in \Pi_{L \setminus \{p\}}$, we define $x'_{\pi p} = x_{\pi' p}$ where $x_{\pi' p}$ minimizes $u_p(x_{\pi' p})$ over all $\pi' \in \Pi_{L\setminus \{p\}}$ such that $\pi$ and $\pi'$ define the same set.
	Moreover, let $x'_\varnothing = x_\varnothing$.
	Clearly, $x'_\pi \in\X^\textsc{CE}_{\pi \cup F}$ for all $\pi \in \Pi_L$ (as each $x'_{\pi}$ is set equal to an $x_{\pi'}$ such that $\pi$ and $\pi'$ correspond to the same set of leaders who performed \textsc{Opt-Out}).
	Moreover, it is easy to check that $\bx'$ is perfectly stable, as follows.
	Let us consider some $p \in L$ and $\pi \in \Pi_{L \setminus \{p\}}$.
	%
	By definition, for every $q \in L \setminus \pi p$, it holds $u_q(x'_{\pi p}) = u_q(x_{\pi' p})$, for some $\pi' \in \Pi_{L \setminus \{p\}}$.
	Moreover, $u_q(x_{\pi' p}) \geq u_q(x_{\pi' p q})$ by stability of $\bx$, and $u_q(x_{\pi' p q}) \geq u_q(x'_{\pi'' p q})$ for some $\pi'' \in \Pi_{L \setminus \{p\}}$.
	Finally, by definition of $\bx'$, we have that $u_q(x'_{\pi'' p q}) = u_q(x'_{\pi p q})$, which shows that $u_q(x'_{\pi p}) \geq u_q(x'_{\pi p q})$.
	Since this holds for any $p \in L$ and $\pi \in \Pi_{L \setminus \{p\}}$, we conclude that $\bx' \in \bX^\textsc{PS}$.
	%
	%
%
\end{proof}

\theoremten*
\begin{proof}
	We build an $\bx =[x_\pi] \in \bX^\textsc{SCE-PA}$ that maximizes $f_\lambda$ by using a stability oracle $\Or$.
	For every $p \in L$ and $\pi\in \Pi_{L \setminus \{p\}}$ with $\pi p=L$, we let $x_{\pi p} = \Or(G,{c}_p,\varnothing, \varnothing)$.
	Otherwise, whenever $\pi p \neq L$, letting $\pi' =\pi p$, we set $x_{\pi'}=\Or(G,{c}_p,L \setminus \pi', \{x_{\pi' q}\}_{q \in L \setminus \pi'})$.
	Moreover, $x_\varnothing=\Or(G,{c}_\lambda,L, \{x_{\{p\}}\}_{p \in L})$.
	Notice that $x_{\pi p} = x_{\pi' p}$ for every $p \in L$ and $\pi,\pi' \in \Pi_{L \setminus \{p\}}$ with $\pi$ and $\pi'$ defining the same set.
	Thus, the number of queries to $\Or$ is 
	$
		\sum_{i =1}^{|L|} |L| \binom{|L|-1}{i-1} +1= |L| 2^{|L|-1} + 1 
	$.
	Clearly, by definition of $\Or$, all the incentive constraints of Eq.~\eqref{eq:incentive_ce} are satisfied.
	Furthermore, it is easy to check that $x_\pi$ is stable for every $\pi \in \Pi_L$.
	As a result, we can conclude that $\bx \in \bX^\textsc{PS}$.
	Now, we prove that $\bx$ maximizes the objective $f_\lambda$ over the set $\bX^\textsc{PS}$.
	This also proves the efficiency of $\bx$, and, thus, $\bx \in \bX^\textsc{SCE-PA}$.
	By contradiction, suppose there exists another $\bx' = [x'_\pi] \in \bX^\textsc{SCE-PA}$ with objective value $f'_\lambda > f_\lambda$. 
	Three cases are possible:
	\begin{itemize}
		\item there exist $p \in L$ and $\pi \in \Pi_{L \setminus \{p\}}$ with $\pi p =L$ such that $u_p(x'_{\pi p}) < u_p(x_{\pi p})$;
		\item there exist $p \in L$ and $\pi \in \Pi_{L \setminus \{p\}}$ with $\pi p  \neq L$ such that $u_p(x'_{\pi p}) < u_p(x_{\pi p})$ and, letting $\pi' = \pi p$, $u_q(x'_{\pi' q}) \geq u_q(x_{\pi' q})$ for all $q \in L \setminus \pi'$;
		\item $u_p(x'_{\{p\}}) \geq u_p(x_{\{p\}})$ for all $p \in L$.
	\end{itemize}
	All the three cases contradict the correctness of $\Or$.
\end{proof}

\propositionsix*
\begin{table}
	{\renewcommand{\arraystretch}{1.1}\setlength{\tabcolsep}{1.7pt}\begin{tabular}{cc|c|c|}
		& \multicolumn{1}{c}{} & \multicolumn{1}{c}{$s_{2,1}$}  & \multicolumn{1}{c}{$s_{2,2}$} \\\cline{3-4}
		& $s_{1,1}$ & $1,2,0$ & $0,1,0$ \\\cline{3-4}
		& $s_{1,2}$ & $1,2,0$ & $0,1,0$ \\\cline{3-4}
		& \multicolumn{1}{c}{} & \multicolumn{2}{c}{$s_{3,1}$}\\
	\end{tabular}}
	{\renewcommand{\arraystretch}{1.1}\setlength{\tabcolsep}{1.7pt}\begin{tabular}{cc|c|c|}
		& \multicolumn{1}{c}{} & \multicolumn{1}{c}{$s_{2,1}$}  & \multicolumn{1}{c}{$s_{2,2}$} \\\cline{3-4}
		& $s_{1,1}$ & $2,1,0$ & $2,1,0$ \\\cline{3-4}
		& $s_{1,2}$ & $1,0,0$ & $1,0,0$ \\\cline{3-4}
		& \multicolumn{1}{c}{} & \multicolumn{2}{c}{$s_{3,2}$}\\
	\end{tabular}}
	{\renewcommand{\arraystretch}{1.1}\setlength{\tabcolsep}{1.7pt}\begin{tabular}{cc|c|c|}	
		& \multicolumn{1}{c}{} & \multicolumn{1}{c}{$s_{2,1}$}  & \multicolumn{1}{c}{$s_{2,2}$} \\\cline{3-4}
		& $s_{1,1}$ & $2,2,0$ & $2,2,0$ \\\cline{3-4}
		& $s_{1,2}$ & $2, 2,0$ & $0,0,10$ \\\cline{3-4}
		& \multicolumn{1}{c}{} & \multicolumn{2}{c}{$s_{3,3}$}\\
	\end{tabular}}
	\caption{Three-player normal-form SG showing that, when searching for an optimal SCE-PA, it is necessary to consider the last leader who performed \textsc{Opt-Out} (players 1, 2, and 3 select rows, columns, and matrices, respectively).}
	\label{table:ordering}
\end{table}

\begin{proof}
	Consider the SG in Table~\ref{table:ordering}, with $L=\{1,2,3\}$ and $F = \varnothing$.
	There exists an $\bx = [x_\pi] \in \bX^\textsc{SCE-PA}$ such that $x_\varnothing(s_{1,2},s_{2,2},s_{3,3})=1$, and the same holds for $x_{\{1\}}(s_{1,1},s_{2,2},s_{3,1})$, $x_{\{2,1\}}(s_{1,2},s_{2,2},s_{3,1})$, $x_{\{2\}}(s_{1,2},s_{2,1},s_{3,2})$, and $x_{\{1,2\}}(s_{1,2},s_{2,1},s_{3,2})$.
	Moreover, for every $\pi\in \Pi_L$ including player 3, $x_\pi(s_{1,1},s_{2,1},s_{3,3})=1$.
	Notice that $x_\pi$ depends on the last player who decides to \textsc{Opt-Out} since $x_{\{1,2\}} \neq x_{\{2,1\}}$.
	We show that there is no $\bx' =[x'_\pi] \in \bX^\textsc{SCE-PA}$ where $x_\pi$ does not depend on the last leader to \textsc{Opt-Out} and $x'_\varnothing(s_{1,2},s_{2,2},s_{3,3})=1$.
	Assume, by contradiction, that there exists such $\bx'$.
	If player 1 performs \textsc{Opt-Out}, she will get more than $0$, unless only the strategy profiles $(s_{1,2},s_{2,2},s_{3,3})$, $(s_{1,1},s_{2,2},s_{3,1})$, and $(s_{1,2},s_{2,2},s_{3,1})$ are recommended by $x'_{\{1\}}$. 
	The other strategy profiles providing player 1 with a utility of $0$ cannot be recommended, otherwise incentive constraints of Eq.~\eqref{eq:incentive_ce} are not satisfied. 
	As a result, only strategy profiles $(s_{1,1},s_{2,1},s_{3,2})$ and $(s_{1,1},s_{2,2},s_{3,2})$ are recommended in $x'_{\{1,2\}}$ (otherwise player 2 would have an incentive to \textsc{Opt-Out}).
	Instead, consider the case in which player 2 performs \textsc{Opt-Out}. 
	Since players 1 and 2 are symmetric, $x'_{\{2,1\}}$ can only recommend strategy profiles $(s_{1,1},s_{2,1},s_{3,1})$ and $(s_{1,2},s_{2,1},s_{3,1})$. 
	Thus, $x'_{\{2,1\}}$ must be different from $x'_{\{1,2\}}$, a contradiction.
\end{proof}

\section{Omitted Proofs for Section~\ref{sec:games_main}}\label{sec:games}

%

In this section, we provide a complete proof of Theorem~\ref{thm:compact_games}.

%
%
%
For ease of presentation, we treat $x \in \X$ as an $|S|$-dimensional vector.
Moreover, given $c = [c_p] \in [-1,1]^n$, let $w = [w_s] \in \mathbb{R}^{|S|}$ be a vector with $w_s= \sum_{p \in N} c_p u_p(s)$, and, given a collection of correlated distributions $\{ x_p \}_{p\in L' \subseteq L}$, let $b_p =  \sum_{s \in S} u_p(s) x_p(s)$ for every $p \in L' \subseteq L$.

%
%
%

The solutions returned by $\Or(G,c,L,\{x_p  \}_{p\in L' \setminus L})$ are the optimal solutions to the following LP: 
\begin{equation*}
	\mathfrak{P}\ :\ \left\{\begin{aligned}
	\max \quad & w^T x \\
	\textnormal{s.t.} \quad & U x \geq 0 \\
	&  \mathbf{1}^T x = 1, \,\, x \geq 0, 
	\end{aligned}\right.
\end{equation*}
where $U$ is a matrix of dimensions $C \times |S|$ (with $C = \sum_{p \in N \setminus L} |S_p|^2 + |L|$) encoding the coefficients of the incentive constraints of Eq.~\eqref{eq:incentive_ce} for the players in $N \setminus L$, and those of the additional stability constraints, \emph{i.e.}, for every $p \in L$
$$
	\sum_{s \in S} \left( u_p(s) - b_p \right) x(s) \geq 0. 
$$
We denote with $U_s$ the column of $U$ corresponding to $s \in S$.

%
%
%
%
%
%

We can write the dual of problem $\mathfrak{P}$ as:
\begin{equation*}
\mathfrak{D}\ :\ \left\{\begin{aligned}
\min \quad & t \\
\textnormal{s.t.} \quad & U^T z + w \leq t \mathbf{1} \\
&  z \geq 0,
\end{aligned}\right.
\end{equation*}
where $z = [z_{s_p, s'_p}^p; z_p] \in \mathbb{R}^C$ is a vector of dual variables: $z_{s_p, s'_p}^p$ for all $p \in N \setminus L$ and $s_p,s'_p \in S_p$, and $z_p$ for all $p \in L$.

%
%
%
%
%
%

A separation problem for $\mathfrak{D}$ asks whether a given pair $(z, t)$ is feasible, and if not, it calls for a hyperplane separating $(z, t)$ from the feasible set. 
Following~\citeauthor{leyton2011ellipsoid}~[\citeyear{leyton2011ellipsoid}], we focus on a restricted form of separation, requiring a violated constraint for infeasible points. 
%
%
Formally:

\begin{definition}[\textsf{Sep}$(z,t)$]\label{def:separation_prob}
	Given a pair $(z, t)$ such that $z \geq 0$, determine if there exists an $s \in S$ such that $(U_s)^T z + w_s > t$; if so output such an $s$.
\end{definition}

%
%
%
%

Notice that, for every $s \in S$, 
\begin{align*}
(U_s)^T z = & \sum_{p \in N \setminus L} \sum_{s'_p \in S_p} z_{s_p,s'_p}^p \left(u_p(s)- u_p(s'_p, s_{-p}) \right) + \\
& + \sum_{p \in L} z_p \left( u_p(s)- b_p \right).
\end{align*}

The following holds:

\setcounter{theorem}{10}
\begin{theorem}\label{thm:poly_ellipsopid}
	If \emph{\textsf{Sep}$(z,t)$} can be solved in polynomial time, then $\Or$ can be computed in polynomial time.
\end{theorem}

\begin{proof}
	Clearly, a polynomial-time algorithm for \textsf{Sep}$(z,t)$ can be used as separation oracle in the ellipsoid method, solving $\mathfrak{D}$ in polynomial time.
	By duality, the optimal objective for $\mathfrak{D}$ is the value $w^T x$ of a solution $x \in \X$ for $\Or$.
	Since we required that separating hyperplanes be constraints for $\mathfrak{D}$, they can be used to compute such solution $x$.
\end{proof}

\begin{corollary}\label{cor:poly_support}
	$\Or$ returns a polynomially-sized $x \in \X$.
\end{corollary}
\begin{proof}
	This is an immediate consequence of the fact that the ellipsoid method, as applied in Theorem~\ref{thm:poly_ellipsopid}, generates a polynomial number of violated constraints.
\end{proof}

Now, we introduce some definitions from~\cite{leyton2011ellipsoid}.
Given a finite game $G$, we let $y = [y_{s_p,s'_p}^p] \in \mathbb{R}^{C'}$ (with $C' = \sum_{p\in N} |S_p|^2$) be a vector indexed by $p\in N$ and $s_p, s'_p \in S_p$.
Moreover, we let $v = [v_p] \in \mathbb{R}^n$ be a vector indexed by $p \in N$.

\begin{definition}
	Given a finite game $G$, a vector $y \in \mathbb{R}^{C'}$ such that $y \geq 0$, and a vector $v \in \mathbb{R}^n$, the \emph{weighted deviation-adjusted utility} for player $p \in N$ in $s \in S$ is:
	$$  
		\hat{u}^p_s(y,v)=v_p u_p(s)+ \sum_{s'_p \in S_p} y_{s_p,s'_p}^p \left( u_p(s)- u_p(s'_p,s_{-p}) \right),
	$$
	and the \emph{weighted deviation-adjusted social welfare} is $\hat{w}_s(y,v)=\sum_{p \in N} \hat{u}^p_s(y)$.
\end{definition}

The following is the formal definition of \emph{weighted deviation-adjusted social welfare problem}.~\footnote{The version proposed by~\citeauthor{leyton2011ellipsoid}~[\citeyear{leyton2011ellipsoid}] adds the additional constraints that $v_p \geq 0$ and $\sum_{p\in N} v_p= 1$.}
\begin{definition}[\textsf{w-DaSW}$(y,v,t)$] \label{problem:2}
	Given a triplet $(y, v, t)$ such that $y \ge 0$, determine if there exists an $s \in S$ such that $\hat{w}_s(y,v) > t$; if so output such an $s$.
\end{definition}

Our main result is the following:

\begin{theorem}\label{thm:sep_reduction}
	\emph{\textsf{Sep}$(z,t)$} reduces to \emph{\textsf{w-DaSW}$(y,v,t)$}.
\end{theorem}

\begin{proof}
	Given $(z,t)$ with $z \geq 0$, asking $(U_s)^T z + w_s > t$ is equivalent to asking
	\begin{align*}
		&\sum_{p \in N \setminus L} \sum_{s'_p \in S_p} y_{s_p,s'_p}^p \left( u_p(s)- u_p(s'_p,s_{-p}) \right) + \\
		& \quad + \sum_{p \in L} \left( c_p + y_p \right) u_p(s)+ \sum_{p \in  N \setminus L} c_p u_p(s) - \sum_{p \in L} y_p b_p >t.
	\end{align*}
	In turn, this is equivalent to solving \textsf{w-DaSW}$(\hat y, \hat v, \hat t)$ with:
	\begin{itemize}
		\item $\hat{y}_{s_p,s'_p}^p=0 $ for all $p \in L, s_p,s'_p \in S_p$;
		\item $\hat{y}_{s_p,s'_p}^p=y_{s_p,s'_p}^p$ for all $ p \in N \setminus L, s_p,s'_p \in S_p$;
		\item $\hat{t}=t + \sum_{p \in L} y_p b_p$;
		\item $\hat{v}_p=c_p+y_p$ for all $p \in L$;
		\item $\hat{v}_p=c_p$ for all $p \in N \setminus L$.
	\end{itemize} 
This concludes the proof.
\end{proof}

%
%

In conclusion, the results in~\cite{leyton2011ellipsoid} together with Theorem~\ref{thm:sep_reduction} prove Theorem~\ref{thm:compact_games}.


%
%
%

\end{document}